\documentclass{article}

\usepackage{geometry}
\usepackage[utf8]{inputenc}

\usepackage[colorlinks]{hyperref}
\usepackage{mathtools}
\usepackage{amsthm}
\usepackage[capitalise]{cleveref}
\usepackage{comment}

\newtheorem{theorem}{Theorem}
\newtheorem{lemma}[theorem]{Lemma}
\newtheorem{proposition}[theorem]{Proposition}
\newtheorem{corollary}[theorem]{Corollary}

\theoremstyle{definition}
\newtheorem{example}[theorem]{Example}
\newtheorem{definition}[theorem]{Definition}
\newtheorem*{notation*}{Notation}
\newtheorem{remark}[theorem]{Remark}

\newtheorem{conjecture}[theorem]{Conjecture}
\usepackage{amssymb}
\usepackage{commath}
\usepackage{enumitem}
\usepackage{tikz}

\newcommand{\GL}{\mathrm{GL}}
\newcommand{\SL}{\mathrm{SL}}

\newcommand{\id}{\mathrm{id}}

\newcommand{\C}{\mathbb{C}}

\newcommand{\R}{\mathbb{R}}
\newcommand{\Q}{\mathbb{Q}}
\newcommand{\Z}{\mathbb{Z}}
\newcommand{\N}{\mathbb{N}}
\newcommand{\bbP}{\mathbb{P}}
\newcommand{\eps}{\varepsilon}
\DeclareMathOperator{\rk}{rk}

\DeclareMathOperator{\supp}{supp}

\newcommand{\sr}{g}
\newcommand{\SR}{G}

\newcommand{\cubedxdydz}[7]{
    \pgfmathparse{#1+#5}\let\dx\pgfmathresult
    \pgfmathparse{#1+#6}\let\dy\pgfmathresult
    \pgfmathparse{#1+#7}\let\dz\pgfmathresult
    \draw[very thick,fill=#2] (#5,#6,\dz) -- (\dx,#6,\dz) -- (\dx,\dy,\dz) -- (#5,\dy,\dz) -- cycle;
    \draw[very thick,fill=#3] (#5,\dy,\dz) -- (\dx,\dy,\dz) -- (\dx,\dy,#7) -- (#5,\dy,#7) -- cycle;
    \draw[very thick,fill=#4] (\dx,#6,\dz) -- (\dx,#6,#7) -- (\dx,\dy,#7) -- (\dx,\dy,\dz) -- cycle;
}

\title{Weighted slice rank and a minimax correspondence to Strassen's spectra}

\date{}

\usepackage{authblk}

\author[1]{Matthias Christandl}
\author[1]{Vladimir Lysikov}
\affil[1]{\small Department of Mathematical Sciences, University of Copenhagen}
\author[3]{Jeroen Zuiddam}
\affil[3]{\small Korteweg--de Vries Institute for Mathematics, University of Amsterdam}

\begin{document}



\newgeometry{top=1cm}

\maketitle

  \vspace*{-1cm}

\begin{abstract}
Structural and computational understanding of tensors is the driving force behind faster matrix multiplication algorithms, the unraveling of quantum entanglement, and the breakthrough on the cap set problem.
Strassen's asymptotic spectra program (FOCS 1986) characterizes optimal matrix multiplication algorithms 
through monotone functionals.
Our work advances and makes novel connections among two recent developments in the study of tensors, namely
\begin{itemize}
    \item the \emph{slice rank} of tensors, a notion of rank for tensors that emerged from the resolution of the cap set problem (Ann.\,Math.\,2017),
    \item and the \emph{quantum functionals} of tensors (STOC 2018), monotone functionals defined as optimizations over moment polytopes.
\end{itemize} 
More precisely, we introduce an extension of slice rank that we call \emph{weighted slice rank} and 
we develop a minimax correspondence between the asymptotic weighted slice rank and the quantum functionals. Weighted slice rank encapsulates different notions of bipartiteness of quantum entanglement.

The correspondence allows us to give a rank-type characterization of the quantum functionals. 
Moreover, whereas the original definition of the quantum functionals only works over the complex numbers, this new characterization can be extended to all fields.
Thereby, in addition to gaining deeper understanding of Strassen's theory for the complex numbers, we obtain a proposal for quantum functionals over other fields. The finite field case is crucial for combinatorial and algorithmic problems where the field can be optimized over.

\smallskip

\begin{center}
\noindent\textbf{R\'esum\'e}
\end{center}

\noindent La compréhension de la théorie des tenseurs, tant dans ses aspects théoriques que numériques, est une motivation majeure dans le développement de plusieurs domaines, dont les algorithmes de multiplication matricielle rapide, mais aussi dans la description de l'enchevêtrement quantique, et dans la percée sur le problème cap-set.
Le programme du spectre asymptotique de Strassen (FOCS 1986) caractérise les algorithmes optimaux de multiplication matricielle via des fonctionnelles monotones.
Notre travail propose de nouveaux liens entre deux développements récents dans l'étude des tenseurs, à savoir : 
\begin{itemize}
    \item le \emph{rang de tranche} des tenseurs, une notion de rang pour les tenseurs qui a vu le jour lors de la résolution du problème cap-set (Ann.\,Math.\,2017) ;
    \item les \emph{fonctionnelles quantiques} des tenseurs (STOC 2018), des fonctionnelles monotones définies comme des optimisations sur des polytopes.
\end{itemize} 
Plus précisément, nous introduisons une généralisation du rang de tranche que nous baptisons \emph{rang de tranche pondéré}, et nous développons une correspondance min-max entre le rang de tranche pondéré asymptotique et les fonctionnelles quantiques. Le rang de tranche pondéré contient des informations sur le caractère bipartite de l'enchevêtrement quantique.

Cette correspondance nous permet de donner une caractérisation des fonctionnelles quantiques en termes de rangs. La définition originelle des fonctionnelles quantiques a lieu sur le corps des nombres complexes, alors que cette nouvelle caractérisation a un sens sur tout corps. Ceci signifie que non seulement nous comprenons maintenant plus profondément la théorie de Strassen sur le corps des nombres complexes, mais nous avons aussi des candidates pour les fonctionnelles quantiques sur d'autres corps. Le cas des corps finis est crucial lorsque l'on étudie des problèmes combinatoires et algorithmiques où il est possible d'optimiser en variant le corps de définition.

\smallskip

\end{abstract}

\restoregeometry


\newpage

\section{Introduction}

The study of the structural and computational aspects of tensors (multi-dimensional arrays of numbers over a field) is deeply connected to key research areas in mathematics, physics and computer science. Examples of these are the research on sunflowers and cap sets in combinatorics~\cite{MR3583358,tao}, understanding quantum entanglement~\cite{MR1804183,li2018tripartite}, tensor networks, secant varieties of Segre varieties in algebraic geometry~\cite{landsberg2012tensors}, scaling algorithms~\cite{DBLP:conf/innovations/BurgisserGOWW18, 8555166}, the complexity of matrix multiplication~\cite{burgisser1997algebraic, blaser2013fast}, and the complexity of boolean functions~\cite{raz2013tensor, blaser2014explicit}.
We advance and make novel connections among two recent developments in the study of tensors:
\begin{itemize}
    \item the \emph{slice rank} of tensors~\cite{tao, sawin,  Blasiak-et-al}, which encapsulates the technique that provided the breakthrough result on the cap set problem in combinatorics,
    \item the \emph{quantum functionals} of complex tensors~\cite{DBLP:conf/stoc/ChristandlVZ18}, a family of multiplicative monotones ---similar in spirit and utility to multiplicative monotones in other settings, like the Lovász theta function in graph theory or the entropy function in information theory---which moved forward the research direction initiated by Strassen on the asymptotic\footnote{In this paper, the term ``asymptotic'' refers to considering large powers of tensors, which is motivated by many applications including the complexity of matrix multiplication, understanding quantum entanglement, and the sunflower and cap set problem in combinatorics. The power is taken under the tensor Kronecker product, which naturally generalizes the matrix Kronecker product.} properties of tensors \cite{Strassen:1986:AST, strassen1987relative, strassen1988asymptotic, strassen1991degeneration}.
\end{itemize}  
To start this off, as a central component of this paper, we introduce a new notion of rank for tensors called \emph{weighted slice rank}. This notion generalizes several existing notions of rank for tensors, most notably the slice rank and the non-commutative rank.
In this paper, we:
\begin{itemize}
    \item prove basic properties of the (asymptotic) weighted slice ranks,
    \item prove, via a minimax argument, a correspondence between the asymptotic weighted slice ranks and the quantum functionals, characterizing each as an optimization problem in terms of the other (over $\C$),
    \item as a consequence, obtain a characterization of asymptotic tripartite-to-bipartite entanglement transformation\footnote{In this task, Alice, Bob and Charlie, given many copies of a joint tripartite quantum state, together try to create as much entanglement as possible between Alice and Bob by stochastic local operations and classical communication (SLOCC) \cite{MR1804183}.} \cite{li2018tripartite}. Equivalently we characterize asymptotic non-commutative rank and asymptotic commutative rank (over $\C$),
    \item propose a notion of quantum functionals over other fields than the complex numbers via our minimax correspondence,
    \item and develop the quantum functionals with respect to the min-entropy instead of the entropy, which leads to a connection to the G-stable rank introduced in \cite{derksen2020gstable}.
\end{itemize}

The minimax correspondence between asymptotic weighted slice rank and the quantum functionals (\cref{thm:f-sr-intro}) has a general form which is of independent interest (\cref{sec:correspondence}).

A crucial gap in previous work is that the usual construction of the quantum functionals is not suitable for extension to other fields. The notion of quantum functionals over finite fields and other fields that we propose is obtained by applying our minimax correspondence to the asymptotic weighted slice ranks, which are defined over any field. We conjecture that these new functions indeed retain the properties of the complex quantum functionals. As evidence of this we prove this conjecture for the subclass of tight tensors over arbitrary fields. Finite fields have shown to be crucial in the study of combinatorial problems like the cap set problem and algorithmic problems like matrix multiplication barriers. 

Our results give a partial resolution to the fundamental structural and computational problem of understanding the \emph{asymptotic tensor restriction problem}. This problem is about the interplay of two tensor notions:
\begin{itemize}
    \item reductions between tensors, called restriction\footnote{A tensor is a \emph{restriction} of another tensor if, as multilinear maps, the first can be obtained from the second by composing with linear maps.
    In the application of tensors to the study of matrix multiplication this notion of restriction indeed amounts to a reduction between computational problems. Restriction will be denoted by $S \leq T$.}---a concept similar to reductions between computational problems in complexity theory,
    \item large powers of tensors\footnote{We take powers under the tensor Kronecker product. This product naturally generalizes the Kronecker product for matrices.}---a regime similar to parallel repetition problems, Shannon capacity, or self-reducible computational problems, like matrix multiplication.
\end{itemize}
The asymptotic tensor restriction problem asks whether, given two tensors, a large power of the first tensor is a restriction of a marginally larger power of the second tensor.
While many tensor parameters are known to be NP-hard to compute, there is an intriguing possibility that the asymptotic tensor restriction problem has an efficiently computable answer. (This would have considerable consequences, as, for example, computing the value of the matrix multiplication exponent is a special case.)
The asymptotic weighted slice ranks as well as the quantum functionals provide limits on which asymptotic restrictions between tensors are possible.

In the remainder of this introduction we will discuss in detail the main concepts and results of the paper. The full results, technical lemmas and proofs we will then discuss in the sections that follow after the introduction.

\textbf{Notation.} In the text we often use minimization over expressions involving $m_i^{1/\xi_i}$ with $m_i \geq 1$ or $\frac{h_i}{\xi_i}$ with $h_i \geq 0$ in the situation where some (but not all) of $\xi_i$ may be zero. If~$\xi_i = 0$, then these expressions are treated as $+\infty$, even if $m_i = 1$ or $h_i = 0$.

\subsection{Weighted slice rank}\label{subsec:wsr}

We introduce a new notion of rank for tensors called the \emph{weighted slice rank}.
In this paper a \emph{tensor} is any element in a tensor space $V_1 \otimes V_2 \otimes V_3$ for finite-dimensional vector spaces~$V_i$ over some field~$K$. In coordinates, relative to the standard product basis $e_i \otimes e_j \otimes e_k$ of $V_1 \otimes V_2 \otimes V_3$, a tensor~$T$ is defined by a three-dimensional array $(T_{i,j,k})_{i,j,k}$ of field elements $T_{i,j,k} \in K$, so that we have the expansion~$T = \sum_{i,j,k} T_{i,j,k}\, e_i \otimes e_j \otimes e_k$. Our results also hold for tensors of order higher than three, but for simplicity of the exposition we will here only talk about tensors of order three.

Weighted slice rank generalizes several important notions of rank for tensors.
To define the weighted slice rank we first need the concept of flattenings and slice decomposition of a tensor.

\begin{definition}[Flattenings]
For any tensor $T = \sum_{i,j,k} T_{i,j,k}\, e_i \otimes e_j \otimes e_k \in V_1 \otimes V_2 \otimes V_3$ we define the \emph{flattenings} of $T$ by \begin{align*}
T_{(1)} = \sum_{i,j,k} T_{i,j,k}\,\, e_i \otimes (e_j \otimes e_k) \in V_1 \otimes (V_2 \otimes V_3),\\
T_{(2)} = \sum_{i,j,k} T_{i,j,k}\,\, e_j \otimes (e_i \otimes e_k) \in V_2 \otimes (V_1 \otimes V_3),\\
T_{(3)} = \sum_{i,j,k} T_{i,j,k}\,\, e_k \otimes (e_i \otimes e_j) \in V_3 \otimes (V_1 \otimes V_2).
\end{align*}
\end{definition}

In coordinates, for a tensor $T = (T_{i,j,k})_{i,j,k}$ the three flattenings of $T$ are given by the three matrices $T_{(1)} = (T_{i,j,k})_{i,(j,k)}$, $T_{(2)} = (T_{i,j,k})_{j,(i,k)}$, and $T_{(3)} = (T_{i,j,k})_{k,(i,j)}$. 

\begin{definition}[Slice decompositions]\label{def:slice-decomp}
  For every $i \in [3] \coloneqq \{1,2,3\}$, we call $T \in V_1 \otimes V_2 \otimes V_3$ an \emph{$i$-slice} if the flattening $T_{(i)}$ has rank one as a matrix.
  For every $r_1, r_2, r_3 \in \N$ we say that \emph{$T$ has a slice decomposition of size $(r_1, r_2, r_3)$} if there are tensors~$T_{i,j} \in V_1 \otimes V_2 \otimes V_3$ ($i \in [3]$, $j \in [r_i]$) such that
  \[
    T = \sum_{j = 1}^{r_1} T_{1,j} + \sum_{j = 1}^{r_2} T_{2,j} + \sum_{j = 1}^{r_3} T_{3,j}
  \]
  and such that each tensor $T_{i,j}$ is an $i$-slice.
\end{definition}

In other words, $T$ has a slice decomposition of size $(r_1, r_2, r_3)$ if and only if there are tensors~$S_1, S_2, S_3$ such that $T = S_1 + S_2 + S_3$ and for every $i \in [3]$ the flattening $(S_i)_{(i)}$ has rank at most~$r_i$.

The \emph{slice rank} of a tensor $T$, introduced by Tao \cite{tao, sawin}, is defined as
\[
\min \{r_1 + r_2 + r_3 \mid \textnormal{$T$ has a slice decomposition of size $(r_1, r_2, r_3)$}\}.
\]
We define the \emph{weighted slice rank} by weighing the numbers $r_1, r_2, r_3$ that appear in the definition of slice rank, as follows. It will become clear later why the particular kind of weighing that we use is the appropriate one.

\begin{definition}[Weighted slice rank]
  Let 
  \[
    \Xi = \{(\xi_1, \xi_2, \xi_3) \mid \xi_1, \xi_2, \xi_3 \geq 0, \max \{\xi_1, \xi_2, \xi_3\} = 1\}.
  \]
  For any $\xi \in \Xi$ we define the \emph{$\xi$-weighted slice rank of $T$} as 
  \[
  S_{\xi}(T) \coloneqq \min \{ r_1^{1 /\xi_1} + r_2^{1/\xi_2} + r_3^{1/\xi_3} \mid \text{$T$ has
  a slice decomposition of size $(r_1, r_2, r_3)$}\}.
  \]
  If any of the $\xi_i$ equals $0$, then we require the value of the corresponding number~$r_i$ to be $0$ and we omit the term $r_i^{1/\xi_i}$.
\end{definition}

The weighted slice rank naturally generalizes the following known notions of rank for tensors.
\begin{itemize}
    \item \textbf{Slice rank:} The weighted slice rank with weight $(1,1,1)$,
    \[
    S_{(1,1,1)}(T) = \min \{r_1 + r_2 + r_3 \mid \textnormal{$T$ has a slice decomposition of size $(r_1, r_2, r_3)$}\},
    \]
    equals the ordinary slice rank of~$T$.
    \item \textbf{Flattening rank:} The matrix ranks of the flattenings $T_{(1)}, T_{(2)}, T_{(3)}$ of $T$ are called the flattening ranks. The flattening ranks are equal to the weighted slice rank with all weight on one of the numbers $r_i$,
    \begin{align*}
    S_{(1,0,0)}(T) &= \min \{r_1 \mid \textnormal{$T$ has a slice decomposition of size $(r_1, 0, 0)$}\},\\
    S_{(0,1,0)}(T) &= \min \{r_2 \mid \textnormal{$T$ has a slice decomposition of size $(0, r_2, 0)$}\},\\
    S_{(0,0,1)}(T) &= \min \{r_3 \mid \textnormal{$T$ has a slice decomposition of size $(0, 0, r_3)$}\}.
    \end{align*}
    \item \textbf{Non-commutative ranks:} The non-commutative ranks of a tensor (see \S\ref{sec:asymptotic-noncommutative-ranks} for the definition) are equal to the weighted slice ranks with weight $1$ on two out of the three numbers~$r_i$,
    \begin{align*}
    S_{(1,1,0)}(T) &= \min \{r_1 + r_2 \mid \textnormal{$T$ has a slice decomposition of size $(r_1, r_2, 0)$}\},\\
    S_{(0,1,1)}(T) &= \min \{r_2 + r_3 \mid \textnormal{$T$ has a slice decomposition of size $(0, r_2, r_3)$}\},\\
    S_{(1,0,1)}(T) &= \min \{r_1 + r_3 \mid \textnormal{$T$ has a slice decomposition of size $(r_1, 0, r_3)$}\}.
    \end{align*}
\end{itemize}

Obviously, slice rank is at most the non-commutative ranks, and the non-commutative ranks are in turn upper bounded by the flattening ranks according to the following Hasse diagram:
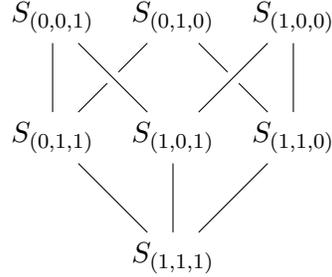
\begin{center}
\begin{tikzpicture}[scale=0.8]
  \node (a) at (-2,2) {$S_{(0,0,1)}$};
  \node (b) at (0,2) {$S_{(0,1,0)}$};
  \node (c) at (2,2) {$S_{(1,0,0)}$};
  \node (d) at (-2,0) {$S_{(0,1,1)}$};
  \node (e) at (0,0) {$S_{(1,0,1)}$};
  \node (f) at (2,0) {$S_{(1,1,0)}$};
  \node (min) at (0,-2) {$S_{(1,1,1)}$};
  \draw (min) -- (d) -- (a) 
  (b) -- (f)
  (e) -- (min) -- (f) -- (c)
  (d) -- (b);
  \draw[preaction={draw=white, -,line width=6pt}] (a) -- (e) -- (c);
\end{tikzpicture}
\end{center}
Generally, for any weightings $\xi, \xi' \in \Xi$, if the inequality $\xi \geq \xi'$ holds elementwise, then we have that $S_\xi \leq S_{\xi'}$ holds on all tensors.

We conclude that the weighted slice ranks naturally generalize and interpolate between the slice rank, the flattening ranks and the non-commutative ranks. Each of these ranks play their own specialized role in various settings. Weighted slice rank for the other possible weightings likewise play their own role, as we will see.

\subsection{Asymptotic weighted slice rank}\label{subsec:intro-awsr}
We are interested in the behaviour of weighted slice rank under taking large powers. The product under which we take powers is the tensor Kronecker product---the natural generalization of the matrix Kronecker product to tensors. For tensors $S \in V_1 \otimes V_2 \otimes V_3$ and $T \in W_1 \otimes W_2 \otimes W_3$ the tensor Kronecker product is a tensor $S \otimes T \in (V_1 \otimes W_1) \otimes (V_2 \otimes W_3) \otimes (V_3 \otimes W_3)$. In coordinates, if $S \in K^{n_1 \times n_2 \times n_3}$ and $T \in K^{m_1 \times m_2 \times m_3}$ then the coefficients of $S\otimes T$ are given by all pairwise products of coefficients of $S$ and coefficients of $T$, 
\[(S\otimes T)_{(i_1, i_2), (j_1, j_2), (k_1, k_2)} = S_{i_1, i_2, i_3} T_{j_1, j_2, j_3}.
\]
The behaviour of weighted slice rank under taking large powers is captured by the asymptotic weighted slice rank.

\begin{definition}[Asymptotic weighted slice rank]\label{def:awsr}
  For $\xi \in \Xi$ we define the \emph{asymptotic $\xi$-weighted slice rank of $T$} as 
  \[
    \SR_{\xi}(T) \coloneqq \limsup_{n \to \infty} S_{\xi}(T^{\otimes n})^{1/n}.
  \]
  We will frequently make use of the logarithm of~$\SR_\xi(T)$ which we denote by $g_\xi(T)$, that is,
  \[
    \sr_{\xi}(T) \coloneqq \log_2 \SR_{\xi}(T) =   \limsup_{n \to \infty} \frac{1}{n} \log_2 S_{\xi}(T^{\otimes n}).
  \]
\end{definition}
Since $G_\xi$ is defined as a lim sup, $G_\xi$ arguably does not capture all of the asymptotic behaviour of the weighted slice rank. For example, we could also consider the corresponding lim inf. However, it will turn out that in the two settings that we will consider (namely, complex tensors and so-called tight tensors), the lim sup and the lim inf coincide, and may equivalently be replaced by a limit.

The weighted slice rank $S_{(1,0,0)}$ is a flattening rank and is therefore multiplicative, therefore $G_{(1,0,0)}(T) = S_{(1,0,0)}(T)$. The same is true for $G_{(0,1,0)}$ and $G_{(0,0,1)}$. For general choices of $\xi \in \Xi$, however, the weighted slice rank~$S_\xi$ will be different from the asymptotic weighted slice rank $G_\xi$. Interesting special cases of the asymptotic weighted slice rank are the asymptotic slice rank $G_{(1,1,1)}$ and the asymptotic non-commutative ranks $G_{(1,1,0)}$, $G_{(1,0,1)}$ and $G_{(0,1,1)}$.

Our first main result is a dual characterization of the asymptotic weighted slice ranks for complex tensors.
This dual characterization is phrased as an optimization problem over a convex polytope called the moment polytope. For every complex tensor $T$, the moment polytope $\Pi(T)$ is defined as follows. Recall that for any tensor $T \in V_1 \otimes V_2 \otimes V_3$ we denote the flattenings in $(V_2 \otimes V_3) \otimes V_1$, $(V_1 \otimes V_3) \otimes V_2$, $(V_1 \otimes V_2) \otimes V_3$ by $T_{(1)}$, $T_{(2)}$, and $T_{(3)}$, respectively. For a nonzero tensor $T$ let $r(T_{(i)})$ be the probability vector obtained by squaring and normalizing the singular values $\sigma_1(T_{(i)}) \geq \sigma_2(T_{(i)}) \geq \cdots$ of the matrix~$T_{(i)}$, that is, 
\[
r(T_{(i)}) = (\sigma_1(T_{(i)})^2, \sigma_2(T_{(i)})^2, \ldots) / \sum_j \sigma_j(T_{(i)})^2.
\]
The moment polytope of a nonzero tensor $T$ is defined as the set of triples of probability vectors
\[
\Pi(T) = \{ (r(S_{(1)}), r(S_{(2)}), r(S_{(3)})) \mid S \in \overline{G\cdot T}, S \neq 0 \}.
\]
In the above definition the set $\overline{G\cdot T}$ is the orbit closure (closure under equivalently the Zariski or Euclidean topology) of $T$ under the action of the group $G = \GL(V_1) \times \GL(V_2) \times \GL(V_3)$. This is the same as the closure of the set of all tensors $S$ such that $S$ is a restriction of $T$.
For the characterization of the asymptotic weighted slice rank we furthermore need the Shannon entropy of a probability vector $p = (p_1, \ldots, p_n)$, which is defined as $H(p) = -\sum_i p_i\log_2 p_i$.
\begin{theorem}\label{thm:sr-polytope-intro}
  Let $T$ be a tensor over $\C$.
  For every $\xi \in \Xi$,
  the limsup in the definition of $\SR_{\xi}(T)$ is a limit
  and
  \[
     \SR_{\xi}(T) = \max_{p \in \Pi(T)} \min_{i\in [3]}  2^{H(p_i) / \xi_i}.
  \]
\end{theorem}
This result extends the result in \cite{DBLP:conf/stoc/ChristandlVZ18} that gave the same dual description only for the asymptotic slice rank (i.e.,~for $\xi = (1,1,1)$).

For tensors over arbitrary fields (i.e.~different from $\C$) we prove a dual description for the asymptotic weighted slice rank for an important subclass of all tensors, namely tensors with a \emph{tight support}\footnote{The \emph{support} of a tensor is the set of all triples of coordinates with nonzero coefficient. We call such a support $S \subseteq [n_1] \times [n_2] \times [n_3]$ \emph{tight} if there are injective maps $u_i : [n_i] \to \Z$ such that for every $s \in S$ it holds that $\sum_i u_i(s_i) = 0$. A tensor in $V_1 \otimes V_2 \otimes V_3$ is called tight if its support is tight for some choice of bases for the vector spaces $V_i$.}. Examples of tensors with a tight support have appeared in the applications to matrix multiplication and combinatorics. In this version of the dual description, the moment polytope~$\Pi(T)$ (which is only available over $\C$) is replaced by the polytope $W(T)$ of triples of marginals of probability vectors on the support of the tensor.

\begin{theorem}\label{thm:sr-tight-intro}
  Let $T$ be a tensor with tight support.
  For every~$\xi \in \Xi$,
  the limsup in the definition of $\SR_{\xi}(T)$ is a limit
  and
  \[
     \SR_{\xi}(T) = \max_{p \in W(T)} \min_{i\in [3]}  2^{H(p_i) / \xi_i}.
  \]
\end{theorem}
This result extends the result in \cite{sawin} and \cite{DBLP:conf/stoc/ChristandlVZ18} that gave the same dual description only for the weighting $\xi = (1,1,1)$, albeit for the broader class of tensors with \emph{oblique} support.

\subsection{Quantum functionals}\label{sec:intro-quantum}
The quantum functionals were introduced in \cite{DBLP:conf/stoc/ChristandlVZ18} to study the asymptotic behaviour of complex tensors, advancing a line of research initiated by Strassen in the context of the arithmetic complexity of matrix multiplication.
The quantum functionals are a family of functions parametrized by the probability simplex
$
    \Theta \coloneqq \{ (\theta_1, \theta_2, \theta_3) \in \R^3 \mid \theta_1, \theta_2, \theta_3 \geq 0,\,  \theta_1 + \theta_2 + \theta_3 = 1 \}.
$
Each quantum functional is defined as an optimization problem over the moment polytope that we defined previously.
Let $\theta \in \Theta$.
  For any nonzero complex tensor $T$ we define
  \[
    F_{\theta}(T) = \max_{p \in \Pi(T)} 2^{ \langle \theta, (H(p_1), H(p_2), H(p_3)) \rangle },
  \]
  where, as before, $H$ denotes the Shannon entropy of a probability vector.
  The functions $F_\theta$ are called the \emph{quantum functionals}.
The main theorem on the quantum functionals that was proved in \cite{DBLP:conf/stoc/ChristandlVZ18} is that for any two complex tensors~$S$ and~$T$ the following properties hold:
  \begin{itemize}
  \item
    Monotonicity: $S \leq T \Rightarrow F_\theta(S) \leq F_\theta(T)$,\footnote{Recall that $S \leq T$ means that $S$ is a restriction of $T$, i.e.~that $S$ can be obtained from $T$ by applying linear maps to the tensor legs.}
  \item
    Additivity: $F_{\theta}(S \oplus T) = F_{\theta}(S) + F_{\theta}(T)$,
  \item
    Multiplicativity: $F_{\theta}(S \otimes T) = F_{\theta}(S) \cdot F_{\theta}(T)$,
  \item
    Normalization: $F_{\theta}(I_n) = n$ for any $n \in \N$ where $I_n = \sum_{i=1}^n e_i \otimes e_i \otimes e_i \in \C^n \otimes \C^n \otimes \C^n$ is the unit tensor.
  \end{itemize}
Tensors parameters that satisfy the above four properties were called \emph{universal spectral points} in~\cite{strassen1988asymptotic}.
In particular, the quantum functionals (and universal spectral points in general) have the following useful application: if an asymptotic inequality $S^{\otimes n} \leq T^{\otimes \alpha n + o(n)}$ holds for some rate $\alpha\in\R_{\geq 0}$, then $F_\theta(S) \leq F_\theta(T)^\alpha$ for all $\theta \in \Theta$. Therefore, contrapositively, if there exists a~$\theta$ such that $F_\theta(S) > F_\theta(T)^\alpha$ then the asymptotic inequality $S^{\otimes n} \leq T^{\otimes \alpha n + o(n)}$ cannot hold. (Considering all currently available results and examples, it is possible that $S^{\otimes n} \leq T^{\otimes \alpha n + o(n)}$ holds if and only if $F_\theta(S) \leq F_\theta(T)^\alpha$ for all $\theta \in \Theta$.)

The above properties of the quantum functionals were used to prove barriers for square matrix multiplication in \cite{christandl_et_al:LIPIcs:2019:10848}, which used the quantum functional with the uniform $\theta = (1/3,1/3,1/3)$. These barriers were concurrently obtained in~\cite{alman:LIPIcs:2019:10834} using slice rank in an important line of work \cite{alman_et_al:LIPIcs:2018:8360, 8555139, blasiak2017groups, Blasiak-et-al} improving and extending the earliest barriers in \cite{10.1145/2746539.2746554}. Barriers for rectangular matrix multiplication were obtained in \cite{christandl2020barriers}, which used the quantum functionals with non-uniform $\theta \in \Theta$.

Our second main result is a correspondence between the quantum functionals and the asymptotic weighted slice ranks. This correspondence allows us to compute the asymptotic weighted slice ranks in terms of the quantum functionals and vice versa.

\begin{theorem}\label{thm:f-sr-intro}
  For every complex tensor $T$ and every $\xi \in \Xi$ we have
  \[
  \SR_{\xi}(T) = \min_{\theta \in \Theta(\xi)} {F_{\theta}(T)}^{1/\langle \theta, \xi\rangle}.
  \]
  where $\Theta(\xi) = \{ \theta \in \Theta \mid \theta_i = 0 \text{ if $\xi_i = 0$}\}$.
  
  Conversely, for every complex tensor $T$ and every $\theta \in \Theta$ we have
  \[
  F_{\theta}(T) = \max_{\xi \in \Xi} \SR_{\xi}(T)^{\langle \theta, \xi\rangle}.
  \]
\end{theorem}

While the quantum functionals are only defined over the complex numbers, the asymptotic weighted slice ranks are defined over any field. We conjecture that via the correspondence in \cref{thm:f-sr-intro} we can extend the quantum functionals to other fields, as follows.

Define for every $\theta \in \Theta$ the function $F^K_\theta$ on tensors $T$ over a field $K$ via
\[
F^K_{\theta}(T) \coloneqq \max_{\xi \in \Xi} \SR_{\xi}(T)^{\langle \theta, \xi\rangle}.
\]
The following properties follow directly from properties of $G_\xi$ as we will see:
\begin{proposition}\label{prop:mon-norm-intro}
For every field $K$ and every $\theta \in \Theta$, the function $F^K_{\theta}$ is monotone and normalized.
\end{proposition}

We make the following conjecture:
\begin{conjecture}\label{conj:intro}
For every field $K$ and every $\theta \in \Theta$, the function $F^K_{\theta}$ is additive and multiplicative.
\end{conjecture}
Thus \cref{conj:intro} (if true), together with \cref{prop:mon-norm-intro}, implies that the $F^K_{\theta}$ are universal spectral points over $K$. We provide two pieces of evidence for \cref{conj:intro}, namely that the conjecture is true for $K = \C$ (since then $F_\theta^\C$ equals the quantum functionals $F_\theta$ for tensors over~$\C$ by \cref{thm:f-sr-intro}) and that the conjecture is true for the subclass of tight tensors over arbitrary fields~$K$.

As another consequence of the dual description in \cref{thm:f-sr-intro} and the super-multiplicativity and super-additivity of the quantum functionals, we will obtain via a general argument the following properties of the asymptotic weighted slice rank.

\begin{corollary}\label{cor:sr-super-intro}
  Over the complex numbers the asymptotic $\xi$-weighted slice rank is:
  \begin{itemize}
  \item Super-multiplicative:
    $\SR_{\xi}(T_1 \otimes T_2) \geq \SR_{\xi}(T_1) \cdot \SR_{\xi}(T_2)$
      \item Super-additive:
    $\SR_{\xi}(T_1 \oplus T_2) \geq \SR_{\xi}(T_1) + \SR_{\xi}(T_2)$.
  \end{itemize}
\end{corollary}

We prove the statement of \cref{cor:sr-super-intro} also for all tight tensors over arbitrary fields via the description in \cref{thm:sr-tight-intro} and the minimax correspondence.

It would be interesting to know if the statement of \cref{cor:sr-super-intro} holds for general tensors over arbitrary fields~$K$.
This would imply that the functions $F_{\theta}^K$ are super-multiplicative and super-additive, and thus prove part of \cref{conj:intro}. Related to this, it would be interesting to have a direct proof of \cref{cor:sr-super-intro} that does not use the super-multiplicativity and super-additivity of the quantum functionals.

In the process of proving the above results we obtain a characterization of an important class of tensors in terms of the asymptotic weighted slice rank. A tensor $T \in V_1 \otimes V_2 \otimes V_3$ is defined to be \emph{semistable}, under the action of the group $\SL(V_1) \times \SL(V_2) \times \SL(V_3)$, if the closure of the orbit $\SL(V_1) \times \SL(V_2) \times \SL(V_3) \cdot T$ does not contain the element 0.

\begin{theorem}\label{cor:ss-charac-intro}
A tensor $T \in V_1 \otimes V_2 \otimes V_3$ with $\dim V_i = m_i$ is semistable under the action of the group $\SL(V_1) \times \SL(V_2) \times \SL(V_3)$ if and only if $\SR_{\xi}(T) = \min_i m_i^{1/\xi_i}$ for every $\xi \in \Xi$.
\end{theorem}

We also prove a version of \cref{cor:ss-charac-intro} for tensors over arbitrary fields, with the condition that~$T$ is semistable replaced by the condition that all powers of $T$ are semistable.

\subsection{Asymptotic non-commutative ranks}\label{sec:asymptotic-noncommutative-ranks}
A tensor $T \in K^{m_1 \times m_2 \times m_3}$ can be seen as a linear map from $(K^{m_3})^*$ to $K^{m_1 \times m_2}$ or, equivalently, a matrix in $K[x_1, \dots, x_{m_3}]^{m_1 \times m_2}$, each entry of which is a linear form in variables $X = (x_1, \dots, x_{m_3})$.
Two notions of rank for such matrices are defined in~\cite{fortin-reutenauer}: the \emph{commutative rank} $\textrm{crk}(T)$, which is the usual rank of a matrix of linear forms over the field $K(X)$, and the \emph{non-commutative rank}, which is the minimal $r$ such that a matrix of linear forms can be decomposed as a product $A B$ with $A \in K[X]^{m_1 \times r}$ and $B \in K[X]^{r \times m_2}$.\footnote{This is actually the definition of \emph{inner rank} in~\cite{fortin-reutenauer}. Non-commutative rank is equal to inner rank over the ring $K\!\left<X\right>$ of non-commutative polynomials, but for matrices of linear forms the definition with $K[X]$ is equivalent, since no polynomials of degree higher than $1$ are involved.}
The weighted slice rank $S_{(1,1,0)}$ of a tensor $T$ is equal to the non-commutative rank of the corresponding matrix of linear forms~\cite[Thm.~1]{fortin-reutenauer}. Similarly, the weighted slice ranks $S_{(1,0,1)}$ and $S_{(0,1,1)}$ of a tensor $T$ are equal to the non-commutative rank of $T$ seen as a matrix in $K[x_1, \ldots, x_{m_2}]^{m_1 \times m_3}$ and $K[x_1, \ldots, x_{m_1}]^{m_2 \times m_3}$, respectively.\footnote{We may also think of there being three non-commutative ranks, namely one for every flattening of the tensor into a matrix.}

Since commutative rank is super-multiplicative~\cite{li2018tripartite}, the limit $\lim_{n \to \infty} \mathrm{crk}(T^{\otimes n})^{1/n}$ exists by Fekete's lemma. Additionally, it is known that $\mathrm{crk}(T) \leq S_{(1,1,0)}(T) \leq 2\mathrm{crk}(T)$~\cite[Cor.2]{fortin-reutenauer}, so asymptotically
\[
  \SR_{(1,1,0)}(T) = \lim_{n \to \infty} S_{(1,1,0)}(T^{\otimes n})^{1/n} = \lim_{n \to \infty} \mathrm{crk}(T^{\otimes n})^{1/n}.
\]

As a special case of \cref{thm:sr-polytope-intro} and \cref{thm:f-sr-intro} we obtain new dual descriptions of the asymptotic non-commutative ranks in terms of the quantum functionals:
\begin{corollary} For every complex tensor $T$ we have that
  \begin{align*}
    \SR_{(1,1,0)}(T) &= \max_{p \in \Pi(T)} \min \{ 2^{H(p_1)},\, 2^{H(p_2)}\} = \min_{\substack{\theta_1, \theta_2 \geq 0 \\ \theta_1 + \theta_2 = 1}} F_{(\theta_1, \theta_2, 0)}(T)^\frac{1}{\theta_1 + \theta_2}
  \end{align*}
\end{corollary}

Non-commutative rank has received much interest in recent years \cite{DBLP:journals/corr/IvanyosQS15b, doi:10.1080/03081087.2017.1337058}.
In quantum information theory, the study of asymptotic non-commutative rank is motivated by the fact that commutative rank is precisely the same as tripartite-to-bipartite entanglement transformations under Stochastic Local Operations and Classical Communication (SLOCC) \cite{PhysRevA.81.052310}, and formulas for asymptotic non-commutative rank have previously been obtained for special cases \cite{li2018tripartite}, \cite[\S 3.4]{jensen2019tensors}.

\subsection{Min-entropy quantum functionals}\label{subsec:intro-min-entropy}
Finally, we explore a new variation on the quantum functionals by replacing the Shannon entropy~$H$ with a smaller notion of entropy called the min-entropy.
The min-entropy of a probability vector $p = (p_1, \ldots, p_n)$ is defined as $H_\infty(p) = -\log_2 \max_i p_i$. Thus the min-entropy is the number of bits required to describe the largest coefficient in a probability vector. It is the smallest among all of the Rényi entropies and in particular satisfies $H_\infty(p) \leq H(p)$ for every $p$.
For $\theta \in \Theta$, we define the \emph{min-entropy quantum functional} as
\[
F_{\theta,\infty}(T) \coloneqq \max_{p \in \Pi(T)} 2^{\langle \theta, (H_\infty(p_1), H_\infty(p_2), H_\infty(p_3))\rangle}.
\]
It follows directly from $H_\infty(p) \leq H(p)$ that $F_{\theta,\infty}(T) \leq F_\theta(T)$. Expanding the singular value definition of the moment polytope $\Pi(T)$ we have
\[
F_{\theta,\infty}(T) = \sup_{S \in G\cdot T} \prod_{i=1}^3 \Biggl( \frac{\sum_j \sigma_j(S_{(i)})^2}{\sigma_1(S_{(i)})^2} \Biggr)^{\!\!\theta_i},
\]
where, as before, $\sigma_1(S_{(i)}) \geq \sigma_2(S_{(i)}) \geq \cdots$ denote the ordered singular values of the matrix $S_{(i)}$. In other words, we may write $F_{\theta,\infty}(T)$ in terms of the Frobenius norm $\norm[0]{\cdot}_F$ and the spectral norm~$\norm[0]{\cdot}_2$ of matrices as
\begin{equation}\label{eq:charac}
F_{\theta,\infty}(T) = \sup_{S \in G\cdot T} \prod_{i=1}^3 \Biggl( \frac{ \norm[0]{S_{(i)}}_F^2}{\norm[0]{S_{(i)}}_2^2} \Biggr)^{\!\!\theta_i}.
\end{equation}
We prove the following basic properties of the min-entropy quantum functionals.

\begin{theorem}\label{th:min-entropy-basic-intro}
The min-entropy quantum functionals are monotone, normalized and super-mul\-ti\-plica\-tive.
\end{theorem}

Next we prove a minimax correspondence to a second family of functions, just like we do in \cref{thm:f-sr-intro} for the quantum functionals. We define for $\xi \in \Xi$ the function
\[
    G_{\xi, \infty}(T) = \max_{p \in \Pi(T)} \min_{i\in [3]}  2^{H_\infty(p_i) / \xi_i}.
\]

\begin{theorem}\label{thm:m-f-sr-intro}
For every $\xi \in \Xi$ and every complex tensor $T$ we have
\[
    \SR_{\xi, \infty}(T) = \min_{\theta \in \Theta} {F_{\theta, \infty}(T)}^{1/\langle \theta, \xi\rangle}.
\]
For every $\theta \in \Theta$ and every tensor $T$ we have
\[
    F_{\theta, \infty}(T) = \max_{\xi \in \Xi} \SR_{\xi, \infty}(T)^{\langle \theta, \xi\rangle}.
\]
\end{theorem}

The function $\SR_{\xi, \infty}(T)$ is similar to the G-stable rank that was introduced in \cite{derksen2020gstable}, but subtly different, as our kind of weighing by $\xi$ is different.

From \cref{th:min-entropy-basic-intro} and \cref{thm:m-f-sr-intro} we prove using a general argument that:

\begin{corollary}\label{cor:m-sr-super-intro}
  For every $\xi \in \Xi$ the function $\SR_{\xi, \infty}$ is monotone, normalized and super-mul\-ti\-plica\-tive.
\end{corollary}

\subsection{Organization of the paper}

This concludes the introduction. In the remaining sections we discuss the results and technical lemmas in detail and provide the proofs for all claims. In \cref{sec:weighted-slice-rank} we discuss the basic properties of the weighted slice rank. In \cref{sec:asymp-weighted-slice-rank} we characterize the asymptotic weighted slice rank for arbitrary complex tensors and for tight tensors over arbitrary fields. In \cref{sec:correspondence} we discuss a general correspondence for what we call dual pair families, of which the quantum functionals and asymptotic weighted slice ranks are an example. In \cref{sec:quantum-functionals} we discuss the correspondence between the quantum functionals and the asymptotic weighted slice ranks, and we discuss the min-entropy quantum functionals.

\section{Weighted slice rank}\label{sec:weighted-slice-rank}
We introduced in \cref{subsec:wsr} for $\xi \in \Xi = \{(\xi_1, \xi_2, \xi_3) \mid \xi_1, \xi_2, \xi_3 \geq 0, \max \{\xi_1, \xi_2, \xi_3\} = 1\}$ the $\xi$-weighted slice rank of a tensor $T$ as the minimum value of $r_1^{1/\xi_1} + r_2^{1/\xi_2} + r_3^{\xi_3}$ such that there is a slice decomposition of $T$ of size $(r_1, r_2, r_3)$.
In this section we will discuss its most basic properties.

\subsection{Alternative characterization}
We begin with an alternative characterization of the notion of a slice decomposition that we defined in \cref{def:slice-decomp}, the straightforward proof of which we leave to the reader.

\begin{lemma}\label{lem:supports}
  A tensor $T \in V_1 \otimes V_2 \otimes V_3$
  has a slice decomposition of size $(r_1, r_2, r_3)$ if and only if there are subspaces $W_i \subseteq V_i$ of dimension $\dim(W_i) = r_i$ such that 
  \[
  T \in W_1 \otimes V_2 \otimes V_3 + V_1 \otimes W_2 \otimes V_3 + V_1 \otimes V_2 \otimes W_3.
  \]
\end{lemma}

In other words, a tensor $T \in V_1 \otimes V_2 \otimes V_3$ has a slice decomposition of size $(r_1, r_2, r_3)$ if and only if there is a choice of bases for the $V_i$ such that the support of $T$ has a $2 \times 2 \times 2$ block structure of the form
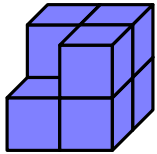
\begin{center}
    \begin{tikzpicture}[line join=round,line cap=round,scale=0.70]
            \colorlet{B}{blue!50}
            \cubedxdydz{1}{B}{B}{B}{0}{2}{1}
    
            \cubedxdydz{1}{B}{B}{B}{1}{1}{1}
            \cubedxdydz{1}{B}{B}{B}{1}{2}{1}
    
            \cubedxdydz{1}{B}{B}{B}{0}{1}{2}
            \cubedxdydz{1}{B}{B}{B}{1}{1}{2}

            \cubedxdydz{1}{B}{B}{B}{1}{2}{2}
    \end{tikzpicture}
\end{center}
where the missing block has dimensions $(\dim(V_1) - r_1) \times (\dim(V_2) - r_2) \times (\dim(V_3) - r_3)$. Thus a slice decomposition corresponds to a block triangularization for tensors.

As a direct consequence of \cref{lem:supports} we have the following alternative characterization of weighted slice rank:
\begin{lemma}
  For every tensor $T \in V_1 \otimes V_2 \otimes V_3$ and $\xi \in \Xi$ the $\xi$-weighted slice rank $\SR_\xi(T)$ is the minimum over $r_1^{1/\xi_1} + r_2^{1/\xi_2} + r_3^{1/\xi_3}$ such that there are subspaces $W_i \subseteq V_i$ such that $\dim(W_i) = r_i$ and 
  \[
  T \in W_1 \otimes V_2 \otimes V_3 + V_1 \otimes W_2 \otimes V_3 + V_1 \otimes V_2 \otimes W_3.
  \]
\end{lemma}

\subsection{Basic properties}\label{subsec:basic-prop}
There are three basic properties of slice decompositions which we will discuss now, followed by a discussion of basic properties of the weighted slice ranks. 
Recall that we defined the restriction preorder $\leq$ on tensors by saying that $T' \leq T$ if and only if there are (not necessarily invertible) matrices $A_i$ such that $T' = (A_1 \otimes A_2 \otimes A_3)\cdot T$.
\begin{lemma}\label{lem:decomp} Let $T$ and $T'$ be tensors.
  \begin{enumerate}[label=\upshape(\roman*)]
      \item\label{lem:decompmon} If $T' \leq T$ and $T$ has a slice decomposition of size $(r_1, r_2, r_3)$, then there are nonnegative integers $r'_i \leq r_i$ such that $T'$ has a slice decomposition of size $(r'_1, r'_2, r'_3)$.
      \item\label{lem:decompsubadd} If $T$ has a slice decomposition of size $(r_1, r_2, r_3)$ and $T'$ has a slice decomposition of size $(r_1', r_2', r_3')$, then
  $T + T'$ has a slice decomposition of size $(r_1 + r_1', r_2 + r_2', r_3 + r_3')$.
      \item\label{lem:decompub}   Every tensor $T \in V_1 \otimes V_2 \otimes V_3$ has a slice decomposition of size $(m_1, 0, 0)$, of size $(0, m_2, 0)$, and of size $(0, 0, m_3)$, where $m_i = \dim(V_i)$.
  \end{enumerate}
\end{lemma}
\begin{proof}
\ref{lem:decompmon} If $S$ is an $i$-slice in the slice decomposition of $T$, and $S' \leq S$, then $S'$ is either an $i$-slice or zero.

\ref{lem:decompsubadd} This follows from adding the two slice decompositions.

\ref{lem:decompub} This follows from the characterization of slice decompositions in \cref{lem:supports}.
\end{proof}

\begin{lemma}\label{lem:sr-basic} Let $T$ and $T'$ be tensors and let $\xi \in \Xi$.
  \begin{enumerate}[label=\upshape(\roman*)]
      \item\label{lem:srmon} If $T' \leq T$, then $S_{\xi}(T') \leq S_{\xi}(T)$.
      \item\label{lem:ub} For every tensor $T \in V_1 \otimes V_2 \otimes V_3$ we have $S_{\xi}(T) \leq \min_i m_i^{1/\xi_i}$ where $m_i = \dim(V_i)$.
      \item\label{lem:srmon-xi} If $\xi \leq \xi'$ elementwise, then $S_{\xi}(T) \geq S_{\xi'}(T)$.
      \item\label{lem:all-max} If a tensor $T \in K^{m \times m \times m}$ has slice rank $S_{(1,1,1)}(T)$ equal to $m$, then $S_{\xi}(T) = m$ for all $\xi \in \Xi$.
  \end{enumerate}
\end{lemma}
\begin{proof}
\ref{lem:srmon} This follows from \cref{lem:decomp} \ref{lem:decompmon}.

\ref{lem:ub} This follows from \cref{lem:decomp} \ref{lem:decompub}.

\ref{lem:srmon-xi} This is clear from the definition of weighted slice rank.

\ref{lem:all-max} From \ref{lem:ub} and \ref{lem:srmon-xi} it follows that $m = S_{(1,1,1)}(T) \leq S_\xi(T) \leq \min_i m^{1/\xi_i} = m$.
\end{proof}

\subsection{Examples}

\begin{example}[Diagonal tensors and the square matrix multiplication tensor]
The diagonal tensors $I_n = \sum_{i = 1}^n e_i \otimes e_i \otimes e_i$ have slice rank $n$~\cite{tao}, so $S_{\xi}(I_n) = n$ for every $\xi \in \Xi$.
Similarly, the slice rank of the square matrix multiplication tensors $\langle n, n, n\rangle = \sum_{i = 1}^n \sum_{j = 1}^n \sum_{k = 1}^n e_{ij} \otimes e_{jk} \otimes e_{ki}$ is~$n^2$ \cite[Rem.~4.9]{Blasiak-et-al}, and so $S_{\xi}(\langle n, n, n\rangle) = n^2$ for every $\xi \in \Xi$.
\end{example}

\begin{example}[Rectangular matrix multiplication tensor]
Consider the rectangular matrix multiplication tensor $\langle n,n,m\rangle = \sum_{i = 1}^n \sum_{j = 1}^n \sum_{k = 1}^m e_{ij} \otimes e_{jk} \otimes e_{ki}$.
The decomposition 
\[
\langle n,n,m\rangle = \sum_{i = 1}^r \sum_{j = 1}^n e_{ij} \otimes \Bigl(\sum_{k = 1}^m e_{jk} \otimes e_{ki}\Bigr) + \sum_{i = r + 1}^n \sum_{k = 1}^m \Bigl(\sum_{j = 1}^n e_{ij} \otimes e_{jk}\Bigr) \otimes e_{ki}
\]
gives the upper bound $S_{(\xi, 1, 1)}(\langle n,n,m\rangle) \leq (nr)^{1/\xi} + m(n - r)$ for every $r$ such that $0 \leq r \leq n$.
If the inequality $n^{1/\xi} < m < n^{1/\xi}(n^{1/\xi} - (n-1)^{1/\xi})$ holds, then this bound is stronger than the general upper bound of~\cref{lem:sr-basic}~\ref{lem:ub}. For example, taking $r = 1$ we get $S_{(0.5,1,1)}(\langle2,2,5\rangle) \leq 9$, which is less than $\min \{4^{1/0.5}, 10, 10\}$.
\end{example}

\section{Asymptotic weighted slice rank}\label{sec:asymp-weighted-slice-rank}

In this section we analyze the asymptotic behaviour of the weighted slice rank when we take large powers of a tensor. More precisely, we study the asymptotic weighted slice ranks $\SR_{\xi}$ defined in~\cref{def:awsr}.
This section is divided into three parts corresponding to three regimes:
\begin{enumerate}
    \item In \cref{subsec:semistable} we give a characterization of the tensors for which the asymptotic weighted slice rank is maximal. This characterization is in terms of semistable tensors. This part works for tensors over arbitrary fields. An important tool that we introduce in this part is a semistability test based on work of Kempf, which we will use in the other parts of this section.
    \item In \cref{subsec:complex} we go deeper than just characterizing maximality and give a description of the value of the asymptotic weighted slice rank in terms of moment polytopes. This part works for tensors over the complex numbers.
    \item Finally, in \cref{subsec:tight} we work over an arbitrary field again and give an upper bound on the value of the asymptotic weighted slice rank in terms of the polytope of probability distributions on the support of the tensors, the torus version of the moment polytope. For the subclass of tensors with tight support (still over arbitrary fields) we prove that this upper bound is tight.
\end{enumerate}
\cref{subsec:complex} and \cref{subsec:tight} each rely on a decomposition for powers of tensors, namely
\begin{itemize}
    \item the Schur--Weyl decomposition and
    \item the weight decomposition.
\end{itemize}
The first is precisely described by the moment polytope and the second by its torus version.

\subsection{Semistable tensors over arbitrary fields}\label{subsec:semistable}

In this section we characterize when the asymptotic weighted slice rank is maximal.
Let $K$ be an arbitrary field and let $V_1, V_2, V_3$ be vector spaces over $K$.
For any tensor $T \in V_1 \otimes V_2 \otimes V_3$ it holds that $\SR_{\xi}(T) \leq \min_i m_i^{1/\xi_i}$ where $m_i = \dim(V_i)$, by the upper bound on the weighted slice rank (\cref{lem:sr-basic}\ref{lem:ub}). We will characterize for what tensors this upper bound is an equality in terms of a notion called semistability from the field of geometric invariant theory.
This characterization generalizes the connection between semistability and slice rank for cubic tensors (i.e.,~tensors of format $n \times n \times n$) explored in~\cite{Blasiak-et-al,DBLP:conf/innovations/BurgisserGOWW18} to tensors of non-cubic format (i.e.,~tensors of format $m_1 \times m_2 \times m_3$ with the $m_i$'s potentially different).

Semistability can be defined in high generality (for actions of reductive algebraic groups on schemes), but we
are only interested in the natural action of the group $\SL(V_1) \times \SL(V_2) \times \SL(V_3)$ on the tensor space $V_1 \otimes V_2 \otimes V_3$.

\begin{definition}
Let $V_1$, $V_2$, $V_3$ be vector spaces over an algebraically closed field $K$.
A tensor $T \in V_1 \otimes V_2 \otimes V_3$ is called \emph{semistable} if the Zariski closure\footnote{Over $\C$ this coincides with the closure in the Euclidean topology. (See, e.g.,~\cite[Lem.~3.1]{Wallach2017}.)} of the orbit $\SL(V_1) \times \SL(V_2) \times \SL(V_3) \cdot T$ does not contain $0$.
Otherwise the tensor is called \emph{unstable}. \footnote{When $K$ is not algebraically closed, we call a tensor semistable (unstable) if it becomes semistable (unstable) after extending the field to the algebraic closure $\overline{K}$.}
\end{definition}

The following semistability test, which follows from the results of Kempf~\cite{kempf1978instability}, and the proof of which we defer to~\ref{sec:sstest}, gives a simple but powerful necessary condition for a tensor over an algebraically closed field to be semistable.

\begin{lemma}[Semistability test]\label{lem:irrsemis}
Let $V_1, V_2, V_3$ be irreducible representations of a group $\Gamma$ defined over an algebraically closed field $K$.
Consider the induced action of $\Gamma$ on the tensor product $V_1 \otimes V_2 \otimes V_3$.
For any nonzero tensor $T \in V_1 \otimes V_2 \otimes V_3$, if~$T$ is invariant under the action of $\Gamma$, then $T$ is semistable under the action of the group $\SL(V_1) \times \SL(V_2) \times \SL(V_3)$.\footnote{We stress that the semistability of tensors is always taken with respect to the action of the group~$\SL(V_1) \times \SL(V_2) \times \SL(V_3)$ and we will not mention this explicitly anymore. The group $\Gamma$ in \cref{lem:irrsemis} could be any group, as long as its action on the $V_i$ is irreducible.}%
$^,$\footnote{The condition in the semistabiliy test is \emph{sufficient} for semistability but not \emph{necessary} as the following example shows. It can be shown that the direct sum $\langle 2,2,2\rangle \oplus \langle 1\rangle$ is semistable using the fact that this tensor is tight and the ideas of \cref{subsec:tight}. However, by considering the stabilizer it follows that this tensor does not satisfy the condition of the semistability test.}
\end{lemma}

We will use the semistability test in proofs in \cref{subsec:complex} and \cref{subsec:tight}. For now we use the test to give simple explicit examples of semistable tensors, so that the reader has some tensors to work with in the rest of this subsection.

\begin{example}
The matrix multiplication tensor $\langle n_1, n_2, n_3\rangle \in (V_1^* \otimes V_2)\otimes (V_2^* \otimes V_3) \otimes (V_3^* \otimes V_1)$, where $V_i = K^{n_i}$, is semistable, because the matrix spaces $V_1^* \otimes V_2$, $V_2^* \otimes V_3$, $V_3^* \otimes V_1$ are irreducible representations of the group $\GL(V_1) \times \GL(V_2) \times \GL(V_3)$ and the matrix multiplication tensor is invariant under the resulting action of this group, which is sometimes called the sandwiching action. Finally, we apply the semistability test (\cref{lem:irrsemis}) to find that the matrix multiplication tensor is semistable.\footnote{In fact, from the stronger \cref{lem:irrpolys} that we discuss in \ref{sec:sstest} it follows that this tensor is \emph{polystable}, meaning that its $\SL$ orbit is closed. For the connection to slice rank we need semistability and therefore we will discuss polystability only in \ref{sec:sstest}.}
\end{example}

\begin{example}\label{ex:diagonal}
The diagonal tensor $\langle n \rangle \coloneqq \sum_{i=1}^n e_i \otimes e_i \otimes e_i \in K^n \otimes K^n \otimes K^n$ is semistable. An argument using the semistability test is as follows.
Let $p$ be an integer. Denote by $V_p$ the representation of $(K^{\times})^n \rtimes S_n$ on $K^n$ given by
$(t_1, \dots, t_n, \sigma) \cdot e_i = t_i^p e_{\sigma i}$.
If $p \neq 0$, then this representation is irreducible.
Indeed, under the action of the torus $(K^{\times})^n$ the space $V_p$ decomposes into a sum of $1$-dimensional weight spaces spanned by the standard basis vectors.
Therefore, every subspace invariant under $(K^{\times})^n$ is spanned by a subset of the standard basis and the only such subspaces invariant under $S_n$ are $\{0\}$ and $V_p$.
The diagonal tensor $\sum_{i = 1}^n e_i \otimes e_i \otimes e_i \in V_1 \otimes V_1 \otimes V_{-2}$ is invariant under $(K^{\times})^n \rtimes S_n$, so it is semistable by the semistability test (\cref{lem:irrsemis}).\footnote{Also here it follows from the stronger \cref{lem:irrpolys} in \ref{sec:sstest} that this tensor is in fact polystable, which means that its $\SL$ orbit is closed.}
\end{example}

The main result of this section is the following characterization of maximality of the asymptotic weighted slice ranks in terms of semistability.

\begin{theorem}\label{thm:asymp-semis}
Let $T \in K^{m_1 \times m_2 \times m_3}$ be a tensor. All Kronecker powers $T^{\otimes k}$ are semistable if and only if $\SR_{\xi}(T)=\min_i m_i^{1/\xi_i}$ for every $\xi \in \Xi$.
\end{theorem}

Note that if $m_1 = m_2 = m_3 = m$, then $\SR_{\xi}(T) = \min_i m_i^{1/\xi_i}$ for every $\xi \in \Xi$ is equivalent to $\SR_{(1,1,1)}(T) = m$. That is, it suffices to look at $\xi = (1,1,1)$, and we recover the connection between asymptotic rank and semistability proved in~\cite{Blasiak-et-al,DBLP:conf/innovations/BurgisserGOWW18}.
For non-cubic tensors, however, one needs to consider the values $\SR_{\xi}(T)$ for $\xi$ other than $(1,1,1)$.
For example, consider a semistable tensor in $K^{m \times m 
\times m}$ embedded into $K^{m \times m \times (m + 1)}$ (by padding with zeros).
The tensor in the larger space is unstable, but the instability is not detected by the asymptotic slice rank $G_{(1,1,1)}(T)$.

Generally, if $T^{\otimes k}$ is semistable, then $T$ is semistable.
Over $\C$, it can be proven that semistability of $T$ is equivalent to the semistability of all Kronecker powers of $T$~\cite[Lem.~6.4 (full version)]{DBLP:conf/innovations/BurgisserGOWW18}.
It is not known whether this property also holds over arbitrary fields. In any case, over the complex numbers this fact together with \cref{thm:asymp-semis} gives:

\begin{corollary}\label{cor:ss-charac}
A tensor $T \in \C^{m_1 \times m_2 \times m_3}$ is semistable if and only if $\SR_{\xi}(T) = \min_i m_i^{1/\xi_i}$ for every~$\xi \in \Xi$.
\end{corollary}

We will now work towards proving the main result of this section, \cref{thm:asymp-semis}.
The crucial connection between slice decompositions and semistability that we will use to prove \cref{thm:asymp-semis} is the following sufficient condition for a tensor to be unstable and the corresponding necessary condition for a tensor to be semistable.

\begin{lemma}\label{lem:slice-unstable}
For any tensor $T \in V_1 \otimes V_2 \otimes V_3$, if $T$ has a slice decomposition of size $(r_1, r_2, r_3)$ such that $r_i < \frac13 \dim V_i$ for every $i\in [3]$, then $T$ is unstable. 

Phrased contrapositively, for any tensor $T \in V_1 \otimes V_2 \otimes V_3$, if $T$ is semistable, then for every slice decomposition of $T$ of size $(r_1, r_2, r_3)$, there is an $i \in [3]$ such that $r_i \geq \tfrac13 \dim V_i$.
\end{lemma}
\begin{proof}
By \cref{lem:supports} there are subspaces $W_i \subset V_i$ with dimensions $\dim W_i = r_i < \frac{\dim V_i}{3}$ such that the tensor $T$ lies in $W_1 \otimes V_2 \otimes V_3 + V_1 \otimes W_2 \otimes V_3 + V_1 \otimes V_2 \otimes W_3$.
For each of these subspaces choose a complement $U_i$, so that $V_i = W_i \oplus U_i$. Denote $q_i = \dim U_i$. Note that $q_i > 2r_i$.

We first assume that all $r_i \neq 0$.
Consider a $1$-parameter subgroup $\lambda(t) = (\lambda_1(t), \lambda_2(t), \lambda_3(t))$ in $\SL(V_1) \times \SL(V_2) \times \SL(V_3)$  where $\lambda_i(t) = t^{\frac{q_i}{r_i} r_1 r_2 r_3} \id \oplus t^{-r_1 r_2 r_3} \id$ with respect to the decomposition $V_i = W_i \oplus U_i$.

The decompositions $V_i = W_i \oplus U_i$ induce a direct sum decomposition of $V_1 \otimes V_2 \otimes V_3$ into eight summands.
The space $W_1 \otimes V_2 \otimes V_3 + V_1 \otimes W_2 \otimes V_3 + V_1 \otimes V_2 \otimes W_3$ containing $T$ is the direct sum of seven out of these eight summands, the missing summand being $U_1 \otimes U_2 \otimes U_3$.
Because $q_i > 2r_i$, on each of these seven summands $\lambda(t)$ acts with a positive power of $t$.
Therefore the curve $\lambda(t) \cdot T$ contains zero in its closure and $T$ is unstable.

If some of the $r_i$ are zero, then we can give a similar construction where $\lambda(t)$ does not act on factors for which $r_i = 0$.
For example, if $r_3 = 0$ we consider the $1$-parameter subgroup $\lambda(t) = (t^{q_1 r_2} \id \oplus t^{-r_1 r_2} \id, t^{r_1 q_2} \id \oplus t^{-r_1 r_2} \id, \id)$.
\end{proof}

Using the necessary condition for semistability in \cref{lem:slice-unstable} we prove the following lower bound on the weighted slice rank.

\begin{lemma}\label{cor:slice-semis}
Let $\xi \in \Xi$. For $T \in K^{m_1 \times m_2 \times m_3}$, if $T$ is semistable, then $S_{\xi}(T) \geq c_{\xi} \min_i m_i^{1/\xi_i}$ for some constant $c_{\xi} > 0$ that depends only on $\xi$.
\end{lemma}
\begin{proof}
Suppose that the value of $S_{\xi}(T^{\otimes n})$ is attained by a slice decomposition of $T$ of size $(r_1, r_2, r_3)$.
By \cref{lem:slice-unstable} there is an $i\in [3]$ such that $r_i \geq \frac{m_i}{3}$.
Therefore
\[S_{\xi}(T) = r_1^{1/\xi_1} + r_2^{1/\xi_2} + r_3^{1/\xi_3} \geq r_i^{1/\xi_i} \geq 3^{-1/\xi_i} m_i^{1/\xi_i} \geq 3^{-1/\min_i \xi_i} \min_i m_i^{1/\xi_i}.
\]
This proves the claim for $c_\xi = 3^{-1/\min_i \xi_i}$.
\end{proof}

By applying \cref{cor:slice-semis} to powers of tensors we get the following sufficient condition for the asymptotic weighted slice rank to be maximal.

\begin{lemma}\label{cor:slice-semis-asymp}
Let $\xi \in \Xi$. For any tensor $T \in K^{m_1 \times m_2 \times m_3}$, if for every $k \in \N$ the tensor $T^{\otimes k}$ is semistable, then $G_{\xi}(T) = \min_i m_i^{1/\xi_i}$.
\end{lemma}
\begin{proof}
By~\cref{cor:slice-semis} we have $S_{\xi}(T^{\otimes k}) \geq c_{\xi} \min_i m_i^{k/\xi_i}$
and asymptotically $G_{\xi}(T) \geq \min_i m_i^{1/\xi_i}$.
\end{proof}

Finally, we prove a necessary condition for the asymptotic weighed slice ranks to be maximal. Combined with the above sufficient condition, this will prove \cref{thm:asymp-semis}.

\begin{lemma}\label{lem:asymp-max-implies-semis}
  Let $T \in K^{m_1 \times m_2 \times m_3}$.
  If $G_{\xi}(T) = \min_i m_i^{1/\xi_i}$ for every $\xi \in \Xi$, then for every $k \in \N$ the tensor $T^{\otimes k}$ is semistable.
\end{lemma}

\begin{proof}
Suppose at least one of the $m_i$ is strictly larger than 1.
Assume that some power $T^{\otimes q}$ is unstable. 
By~\cite[Thm.~4.10]{Blasiak-et-al} the powers $T^{\otimes qk}$ have slice decompositions of size $(r_1, r_2, r_3)$ with $r_i \leq m_i^{qk}e^{-\alpha k}$ where $\alpha > 0$ is the instability of $T^{\otimes q}$ as defined in~\cite{Blasiak-et-al}.
Denote $m = \max_i m_i$ and take $\xi_i = \log_m m_i$ so that $m_i^{1/\xi_i} = m$.
Clearly, $\xi = (\xi_1, \xi_2, \xi_3)$ lies in $\Xi$.
The corresponding weighted slice rank can be bounded as $$S_{\xi}(T^{\otimes qk}) \leq r_1^{1/\xi_1} + r_2^{1/\xi_2} + r_3^{1/\xi_3} \leq m^{qk} \left(e^{-\frac{\alpha}{\xi_1} k} + e^{-\frac{\alpha}{\xi_2} k} + e^{-\frac{\alpha}{\xi_3} k}\right) \leq  3 \left(m e^{-\frac{\alpha}{q \min_i \xi_i}}\right)^{qk}.$$
It follows that the asymptotic $\xi$-weighted slice rank cannot have the maximal possible value $m$.

Consider now the case when $m_1 = 1$.
In this case, the tensor can be presented as $T = e_1 \otimes M$ for some matrix $M \in K^{m_2 \times m_3}$.
Assume $M$ is nonzero and $\rk M = r$.

Every slice decomposition for $T$ of size $(0, r_2, r_3)$ can be transformed into a rank decomposition for $M$ with $r_2 + r_3$ summands by forgetting the first tensor factor.
Therefore $S_{(0,1,1)}(T) \geq \rk M$.
On the other hand, every rank $r$ decomposition for $M$ can be transformed into a slice decomposition for $T$ of size $(0, r, 0)$ or $(0, 0, r)$ by tensoring one of the factors with $e_1$, which gives an upper bound $S_{\xi}(T) \leq r$ if $\xi_1 = 0$.
Since $S_{(0,1,1)}(T) \leq S_{\xi}(T)$ if $\xi_1 = 0$, we have $S_{\xi}(T) = r$ if $\xi_1 = 0$.
If $\xi_1 > 0$, then $S_{\xi}(T) = 1$ from the decomposition $T = e_1 \otimes M$.

Since matrix rank is multiplicative, we have the same values for the asymptotic weighted slice ranks: $G_{\xi}(T) = 1$ if $\xi_1 > 0$ and $G_{\xi}(T) = r$ if $\xi_1 = 0$.
The equality $G_{\xi}(T) = \min_i m_i^{1/\xi_i}$ holds if and only if $r = m_2 = m_3$.
If this is the case, the $\SL_1 \times \SL_{m_2} \times \SL_{m_3}$ orbit of $T$ consists of tensors of the form $e_1 \otimes M'$ with $\det M'  = \det M$, it is closed and does not contain $0$, so $T$ is semistable.
\end{proof}

\begin{proof}[Proof of \cref{thm:asymp-semis}]
Combine the sufficient condition in \cref{cor:slice-semis-asymp} and the necessary condition in \cref{lem:asymp-max-implies-semis}.
\end{proof}

\subsection{Arbitrary tensors over the complex numbers}\label{subsec:complex}

In this section we characterize, for tensors over $\C$, the value of the asymptotic weighted slice rank in terms of moment polytopes. Recall that the moment polytope of a tensor $T$ is denoted by $\Pi(T)$, and that~$H$ denotes the Shannon entropy of probability vectors. We gave a definition of moment polytopes in terms of singular values in \cref{subsec:intro-awsr} and we will see another of the standard descriptions soon. The characterization that we prove is as follows.

\begin{theorem}[\cref{thm:sr-polytope-intro} in the introduction]\label{thm:sr-polytope}
  Let $V_1$, $V_2$ and $V_3$ be vector spaces over $\C$
  and let $T \in V_1 \otimes V_2 \otimes V_3$ be a tensor.
  For every $\xi \in \Xi$ we have
  \[
     \SR_{\xi}(T) = \max_{p \in \Pi(T)} \min_{i\in [3]}  2^{H(p_i) / \xi_i}.
  \]
\end{theorem}

To carry out the proof of \cref{thm:sr-polytope} we will use the following notation and fundamental results.
  We write $\lambda \vdash_m k$ if $\lambda$ is a partition of $k$ into $m$ parts.
  Such partitions are represented by non-decreasing tuples in $\N^m$.
  We omit $k$ or $m$ if its value is not restricted.
  We write
  $(\lambda_1, \lambda_2, \lambda_3) \vdash_{(m_1, m_2, m_3)} k$ if $\lambda_i \vdash_{m_i} k$ for every~$i$.
  Again we may omit some of the parameters when they are not restricted.

  In characteristic $0$, the irreducible representations of the symmetric group $S_k$ (Specht modules) 
  correspond to partitions $\lambda \vdash k$.
  We denote the representation corresponding to a partition $\lambda$ by $[\lambda]$.
  The irreducible polynomial representations of $\GL_n$ (Schur modules)
  correspond to partitions $\lambda \vdash_n$ and are denoted by $\{\lambda\}$.

\begin{theorem}[Schur--Weyl duality]
  As a $\GL_m \times S_k$-representation, the space $(\C^m)^{\otimes k}$ decomposes into
  \[
    (\C^m)^{\otimes k} \cong \bigoplus_{\lambda \vdash_m k} \{\lambda\} \otimes [\lambda].
  \]
\end{theorem}

We will need the following bounds on the dimensions of $\{\lambda\}$ and $[\lambda]$.

\begin{lemma}[\cite{Christandl_2005}]\label{thm:reprdim-bound}
For any partition $\lambda \vdash_m k$ we have the following bounds on the dimension of Schur and Specht modules:
\begin{gather*}
\dim \{\lambda\} \leq (k + 1)^{\frac{m(m-1)}{2}} \\
\frac{k!}{\prod_i (\lambda_i + m - i)!} \leq \dim [\lambda] \leq \frac{k!}{\prod_i \lambda_i!}.
\end{gather*}
\end{lemma}
In particular, it follows from \cref{thm:reprdim-bound} that for any partition $\lambda = (\lambda_1, \ldots, \lambda_m) \vdash k$, if we consider the stretched partition $n\lambda = (n\lambda_1, \ldots, n\lambda_m) \vdash nk$, and we let $n$ go to infinity, then the rate of growth of the dimension of the Specht module $[n\lambda]$ is precisely described in terms of the Shannon entropy $H(\lambda/k)$ of the probability vector $\lambda/k = (\lambda_1/k, \ldots, \lambda_m/k)$ as follows.
\begin{lemma}\label{cor:reprdim-bound}
For any $\lambda \vdash k$ we have $\dim [n\lambda] \leq 2^{H(\frac{\lambda}{k})kn}$ for any $n \in \N$ and $2^{H(\frac{\lambda}{k})kn - o(n)} \leq \dim [n\lambda]$ as $n \to \infty$.
\end{lemma}

Finally, regarding dimensions, we will need the following inequality for the dimension of the Specht module $[\lambda + \mu]$ where as usual $\lambda + \mu$ is the partition $(\lambda_1 + \mu_1, \ldots, \lambda_m + \mu_m)$.

\begin{lemma}\label{lem:reprdim-sum}
For any partitions $\lambda$ and $\mu$ we have $\dim [\lambda + \mu] \geq \dim[\lambda] \dim [\mu]$.
\end{lemma}
\begin{proof}
For any partition $\nu \vdash k + \ell$ the $S_{k + \ell}$ representation $[\nu]$ decomposes as a $S_{k} \times S_\ell$ representation~as
\[
[\nu]\!\downarrow^{S_{k + \ell}}_{S_k \times S_\ell} \cong \bigoplus_{\lambda \vdash k, \mu \vdash l} c^{\nu}_{\lambda \mu}\, [\lambda] \otimes [\mu]
\]
where the $c^{\nu}_{\lambda \mu}$ are the so-called Littlewood--Richardson coefficients~\cite[Ex.~4.43]{Fulton-Harris}.
In particular, for $\nu = \lambda + \mu$ it is known that $c^{\nu}_{\lambda \mu} = 1$~\cite[Ex.~15.24]{Fulton-Harris}. Therefore, supposing $\lambda \vdash k$ and $\mu \vdash \ell$, we have that $[\lambda] \otimes [\mu]$ is contained in $[\lambda + \mu]\!\downarrow^{\smash{S_{k + \ell}}}_{S_k \times S_\ell}$. It follows that $\dim [\lambda + \mu] \geq \dim [\lambda] \dim [\mu]$.
\end{proof}

The Schur--Weyl decompositions of the three representations $(\C^{m_1})^{\otimes k}$, $(\C^{m_2})^{\otimes k}$ and $(\C^{m_3})^{\otimes k}$ induce a decomposition of their tensor product $(\C^{m_1})^{\otimes k} \otimes (\C^{m_2})^{\otimes k} \otimes (\C^{m_3})^{\otimes k}$ as a direct sum of spaces $(\{\lambda_1\}_{\GL_{m_1}} \otimes [\lambda_1]) \otimes (\{\lambda_2\}_{\GL_{m_2}} \otimes [\lambda_2]) \otimes (\{\lambda_3\}_{\GL_{m_3}} \otimes [\lambda_3])$.
For every tensor $T \in \C^{m_1 \times m_2 \times m_3}$ and every partition triple $\lambda \vdash_{(m_1, m_2, m_3)} k$ we denote by $T^{\lambda}$ the projection of $T^{\otimes k}$ onto the direct summand $(\{\lambda_1\}_{\GL_{m_1}} \otimes [\lambda_1]) \otimes (\{\lambda_2\}_{\GL_{m_2}} \otimes [\lambda_2]) \otimes (\{\lambda_3\}_{\GL_{m_3}} \otimes [\lambda_3])$. The nonzero summands $T^{\lambda}$ satisfy the following semigroup property.

\begin{theorem}[{\cite[\S 3c]{Walter1205}}]\label{thm:orbit-semigroup}
  Let $T \in \C^{m_1 \times m_2 \times m_3}$ be a tensor.
  The set of all partitions $\lambda \vdash_{(m_1, m_2, m_3)}$ such that $T^{\lambda} \neq 0$ has the semigroup property: if $T^{\lambda} \neq 0$ and $T^{\mu} \neq 0$, then $T^{\lambda + \mu} \neq 0$.
\end{theorem}

The semigroup property (\cref{thm:orbit-semigroup}) implies that $\{ (\frac{\lambda_1}{k}, \frac{\lambda_2}{k}, \frac{\lambda_3}{k}) \mid \lambda \vdash_{(m_1, m_2, m_3)} k, T^{\lambda} \neq 0\}$ is a convex subset of $\Q^{m_1 + m_2 + m_3}$.
It is a remarkable fact that this set coincides with the set of rational points in the moment polytope~$\Pi(T)$ defined in the introduction. This is the representation-theoretic definition of the moment polytope.
This was proven in high generality in~\cite{Brion-moment-86} and the tensor case is explored in~\cite{Walter1205}.
The points of the moment polytope are tuples of $m_1 + m_2 + m_3$ real numbers
divided into three parts with~$m_i$ numbers in each part.
The three parts of the tuple $p$ will denoted by $p_1$, $p_2$ and~$p_3$.
We denote by $H(p)$ the tuple of entropies $(H(p_1), H(p_2), H(p_3))$.

We will now work towards the proof of the main theorem of this subsection, \cref{thm:sr-polytope}.
The proof of \cref{thm:sr-polytope} is divided into several lemmas.
Before we start, let us establish some notation.
Let $m_i = \dim V_i$ and $m = (m_1, m_2, m_3)$. Let $T \in \C^{m_1 \times m_2 \times m_3}$.
Recall that for a partition triple $\lambda \vdash_m k$ the tensor $T^{\lambda}$ is the projection of $T^{\otimes k} \in V_1^{\otimes k} \otimes V_2^{\otimes k} \otimes V_3^{\otimes k}$ onto the $\GL(V_1) \times \GL(V_2) \times \GL(V_3) \times S_k$-subrepresentation $(\{\lambda_1\}_{\GL(V_1)} \otimes [\lambda_1]) \otimes (\{\lambda_2\}_{\GL(V_2)} \otimes [\lambda_2]) \otimes (\{\lambda_3\}_{\GL(V_3)} \otimes [\lambda_3])$.

For every $k \in \N$ we define the set $R_k(T) = \{ \lambda \vdash_m k \mid T^{\lambda} \neq 0\} \subset \Z^{m_1 + m_2 + m_3}$.
The semigroup property (\cref{thm:orbit-semigroup}) implies that $R_{k + \ell}(T) \supset R_k(T) + R_\ell(T)$ and the representation-theoretic characterization of the moment polytope gives that $\Pi(T)$ is the closure of $\bigcup_k \frac{R_k(T)}{k}$.
Define $M_{\xi, k}(T) = \max_{\lambda \in R_k(T)} \min_i (\dim [\lambda_i])^{1/\xi_i}$.

The next lemma gives a lower bound on the weighted slice rank of powers of tensors, by using the semistability test (\cref{lem:irrsemis}) and the lower bound on the weighted slice rank of semistable tensors (\cref{cor:slice-semis}).

\begin{lemma}\label{lem:sr-polytope-lower}
  For every tensor $T$ and $\xi \in \Xi$ we have $S_{\xi}(T^{\otimes k}) \geq c_{\xi} M_{\xi, k}(T)$ for some $c_{\xi} > 0$ dependent only on $\xi$.
\end{lemma}
\begin{proof}
Each component tensor $T^{\lambda}$ is obtained from the Kronecker power $T^{\otimes k}$ by applying the projections $V_i^{\otimes k} \to \{\lambda_i\} \otimes [\lambda_i]$ to the three tensor factors, so $T^{\otimes k} \geq T^{\lambda}$.

Since $T^{\otimes k}$ is $S_k$-invariant and the projections are equivariant, each $T^{\lambda}$ is $S_k$-invariant.
Furthermore, because the action of $S_k$ on $\{\lambda_i\}$ is trivial, for each linear form $f_i \colon \{\lambda_i\} \to \C$ the restriction
$
[(f_1 \otimes \id) \otimes (f_2 \otimes \id) \otimes (f_3 \otimes \id)] \, T^{\lambda} \in [\lambda_1] \otimes [\lambda_2]\otimes [\lambda_3]
$ is $S_k$-invariant.

If $\lambda \in R_k(T)$, that is, if $T^{\lambda}$ is nonzero, then there exist $f_i$ such that this restriction is nonzero (this follows by a triple application of a general fact that if $x \in U \otimes V$ is nonzero, then $(f \otimes \id) x \in V$ is nonzero for some $f$).
Denote by $T'$ a nonzero tensor of the form $[(f_1 \otimes \id) \otimes (f_2 \otimes \id) \otimes (f_3 \otimes \id)] \, T^{\lambda}$.
By the semistability test (\cref{lem:irrsemis}) the tensor $T'$ is semistable.
Using~\cref{cor:slice-semis} and the monotonicity of weighted slice ranks we get $S_{\xi}(T^{\otimes k}) \geq S_{\xi}(T') \geq c_{\xi} \min_i (\dim [\lambda_i])^{1/\xi_i}$.
Maximizing over all $\lambda \in R_k$ we get the required lower bound.
\end{proof}

\begin{lemma}\label{lem:sr-polytope-upper}
For all positive integers $m_1, m_2, m_3$ and $\xi \in \Xi$ there exists $\gamma > 0$ such that for every tensor $T \in \C^{m_1 \times m_2 \times m_3}$ and every positive integer $k$ we have $S_{\xi}(T^{\otimes k}) \leq 3 (k+1)^{\gamma} M_{\xi, k}(T)$ .
\end{lemma}
\begin{proof}
We prove the statement for the case when all three $\xi_i$ are nonzero.
The tensor~$T^{\otimes k}$ is the sum of $T^{\lambda}$ for all $\lambda \in R_k(T)$.
By \cref{lem:decomp}\ref{lem:decompub} each summand $T^\lambda$ has slice decompositions of sizes $(\dim [\lambda_1] \dim \{\lambda_1\}_{\GL(V_1)}, 0, 0)$, $(0, \dim [\lambda_2] \dim \{\lambda_2\}_{\GL(V_2)}, 0)$ and $(0, 0, \dim [\lambda_3] \dim \{\lambda_3\}_{\GL(V_3)})$.
From the definition of $M_{\xi, k}(T)$ it follows that
for every $\lambda \in R_k(T)$ there is an index $i$ such that $\dim [\lambda_i] \leq M_{\xi, k}^{\xi_i}$.
By~\cref{thm:reprdim-bound} the dimensions $\dim \{\lambda_i\}_{\GL(V_i)}$ are upper bounded by $(k + 1)^{\alpha}$ for $\alpha = \max_i \frac{m_i(m_i-1)}{2}$.
Therefore, for every $\lambda\in R_k(T)$ the tensor $T^\lambda$ has either a slice decomposition of size $((k + 1)^{\alpha} M_{\xi, k}^{\xi_1}, 0, 0)$ or a slice decomposition of size $(0, (k + 1)^{\alpha} M_{\xi, k}^{\xi_2}, 0)$ or a slice decomposition of size $(0, 0, (k + 1)^{\alpha} M_{\xi, k}^{\xi_3})$.
Additionally, the number of partition triples $\lambda \vdash_m k$ is upper bounded by $(k + 1)^{m_1 + m_2 + m_3}$.
Summing the slice decompositions for all nonzero $T^{\lambda}$, we get a slice decomposition for $T$ of size at most
$((k + 1)^{\beta} M_{\xi, k}^{\xi_1}, (k+1)^{\beta} M_{\xi, k}^{\xi_2}, (k + 1)^{\beta} M_{\xi, k}^{\xi_3})$ where $\beta = m_1 + m_2 + m_3 + \alpha$.
Therefore, the $\xi$-weighted slice rank $S_{\xi}(T^{\otimes k})$ is at most $3 (k + 1)^{\gamma} M_{\xi, k}$ where $\gamma = \frac{\beta}{\min_i \xi_i}$.

The proof for the case when some $\xi_i = 0$ is analogous, except that we do not allow $i$-slices in our slice decompositions.
\end{proof}

\begin{lemma}\label{lem:sr-polytope-limit}
  The sequence $\sqrt[k]{M_{\xi, k}(T)}$ converges when $k$ goes to infinity, and the limit is given by
  \[
  \lim_{k\to\infty} \sqrt[k]{M_{\xi, k}(T)} = \max_{p \in \Pi(T)} \min_i 2^{H(p_i)/\xi_i}.
  \]
\end{lemma}
\begin{proof}
Using $R_{k + \ell}(T) \supseteq R_k(T) + R_\ell(T)$ and~\cref{lem:reprdim-sum} we can prove the following super-multiplicativity property:
\begin{align*}
M_{\xi, k + \ell}(T) & = \max_{\lambda \in R_{k + \ell}(T)} \min_i (\dim [\lambda_i])^{\frac{1}{\xi_i}} \\
& \geq \max_{\mu \in R_k(T), \nu \in R_\ell(T)} \min_i (\dim [\mu_i + \nu_i])^{\frac{1}{\xi_i}} \\
& \geq \max_{\mu \in R_k(T), \nu \in R_\ell(T)} \min_i (\dim [\mu_i] \dim [\nu_i])^{\frac{1}{\xi_i}} \\
& \geq \max_{\mu \in R_k(T)} \min_i (\dim [\mu_i])^{\frac{1}{\xi_i}} \max_{\nu \in R_\ell(T)} \min_i (\dim [\nu_i])^{\frac{1}{\xi_i}} \\
& = M_{\xi, k}(T) M_{\xi, \ell}(T).
\end{align*}
From the upper bound on the dimension of $[\lambda_i]$ in terms of Shannon entropy (\cref{thm:reprdim-bound}) and the inclusion $\frac{R_k(T)}{k} \subseteq \Pi(T)$ we obtain
\[
\sqrt[k]{M_{\xi, k}(T)} = \max_{\lambda \in R_{k}} \min_i (\sqrt[k]{\dim [\lambda_i]})^{1/\xi_i} \leq \max_{\lambda \in R_{k}} \min_i 2^{H(\frac{\lambda_i}{k})/\xi_i} \leq \max_{p \in \Pi(T)} \min_i 2^{H(p_i)/\xi_i}.
\]
The convergence of $\sqrt[k]{M_{\xi, k}(T)}$ now follows from Fekete's lemma applied to the sequence
\[
\log \sqrt[k]{M_{\xi, k}(T)} = \frac{\log M_{\xi, k}(T)}{k},
\]
since it has a super-additive numerator and it is bounded.

On the other hand, using the lower bound $\dim [n\lambda_i] \geq 2^{H(\frac{\lambda_i}{k})kn - o(n)}$ (\cref{thm:reprdim-bound}) we see that
\[
\sqrt[kn]{M_{\xi, kn}(T)} \geq \max_{\lambda \in R_{k}} \min_i 2^{H(\frac{\lambda_i}{k})/\xi_i - o(1)}
\]
and letting $n$ go to infinity we get
\[
\lim_{n \to\infty} 
\sqrt[n]{M_{\xi, n}(T)}
= \lim_{n \to\infty} 
\sqrt[kn]{M_{\xi, kn}(T)}
\geq \max_{\lambda \in R_{k}} \min_i 2^{H(\frac{\lambda_i}{k})/\xi_i}.
\]
Considering a sequence of rational points $\frac{\lambda^{(k)}}{k} \in \Pi(T)$ which converges to the point where the maximum $\max_{p \in \Pi(T)} \min_i 2^{H(p_i)/\xi_i}$ is attained, we see that this lower bound can be made arbitrarily close to the upper bound $\max_{p \in \Pi(T)} \min_i 2^{H(p_i)/\xi_i}$.
\end{proof}

\begin{proof}[Proof of~\cref{thm:sr-polytope}]
Lemmas~\ref{lem:sr-polytope-lower} and~\ref{lem:sr-polytope-upper} prove the bounds
\[
c_{\xi} M_{\xi, k}(T) \leq S_{\xi}(T^{\otimes k}) \leq (k + 1)^{\gamma} M_{\xi, k}(T).
\]
Taking the $k$th root and letting $k$ go to infinity, and using \cref{lem:sr-polytope-limit}, we get the required
\[
\SR_{\xi}(T) = \lim_{k \to\infty} \sqrt[k]{M_{\xi, k}(T)} = \max_{p \in \Pi(T)} \min_i 2^{H(p_i)/\xi_i}.\qedhere
\]
\end{proof}

\begin{remark}
Note that not only have we obtained a characterization of $G_{\xi}(T)$ for every tensor $T$ over $\C$, but from the proof it also follows that the $\limsup$ in the definition of $G_{\xi}(T)$ can be replaced by a limit.
\end{remark}

\subsection{Tight tensors over arbitrary fields}\label{subsec:tight}

In the previous section we proved an upper bound on the asymptotic weighted slice rank for tensors over the complex numbers. Here we prove a similar upper bound for tensors over arbitrary fields. Moreover, we show that this upper bound is optimal for the subclass of \emph{tight} tensors.
The idea is similar to the complex case, but instead of the Schur--Weyl decomposition and moment polytope, which are not available in positive characteristic, we use the weight decomposition and the corresponding polytope.

Let $e_1, \dots, e_m$ be the standard basis of the vector space $K^m$. The tensor products $e_{j_1} \otimes e_{j_2} \otimes e_{j_3}$ form a basis for the tensor product space $K^{m_1 \times m_2 \times m_3}$.
Every tensor $T \in K^{m_1 \times m_2 \times m_3}$ can be written as
\[
  T = \sum_{j_1 = 1}^{m_1} \sum_{j_2 = 1}^{m_2} \sum_{j_3 = 1}^{m_3} t_{j_1 j_2 j_3}\, e_{j_1} \otimes e_{j_2} \otimes e_{j_3}.
\]
The \emph{support} of $T$, denoted by $\supp(T)$, is the set of all triples $(j_1, j_2, j_3)$ for which $t_{j_1 j_2 j_3} \neq 0$.

For any $k$-tuple $s \in [m]^k$ we define the vector $e_s \coloneqq e_{s_1} \otimes e_{s_2} \otimes \dots \otimes e_{s_k}$.
These vectors form a basis of $(K^m)^{\otimes k}$.
We define the \emph{type} of $s$ as the $m$-tuple $w \in \N^m$ with coefficients $w_j = \# \{ i \mid s_i = j \}$ for~$j \in [m]$. That is, for every $j \in [m]$ the value of $w_j$ is the number of appearances of the element~$j$ in the tuple~$s$.
The tuple~$s$ contains $k$ elements in total, so $\sum_{i = 1}^m w_i = k$. 
The set of possible types of elements in $[m]^k$ is the set of all ordered partitions of $k$ into $m$ parts.
We may think of these types as rational probability distributions $\frac{w}{k}$ on $[m]$.
Let $[[w]] \subseteq K^m$ be the subspace spanned by the basis elements $e_s$ such that $s \in [m]^k$ is of type~$w$.
The dimension of $[[w]]$ equals the multinomial coefficient $\binom{k}{w} \coloneqq k! / (w_1! \cdots w_m!)$.
The subspaces $[[w]]$ are $S_k$-subrepresentations of $(K^m)^{\otimes k}$ called \emph{permutation modules}.
The space~$(K^m)^{\otimes k}$ decomposes as the direct sum of~$[[w]]$ over all possible types $w$.

We write $(w_1, w_2, w_3) \vDash_{(m_1, m_2, m_3)} k$ if $w_i$ is the type of an element of $[m_i]^{k}$ for $i = 1,2,3$.
For every tensor $T \in K^{m_1 \times m_2 \times m_3}$
the Kronecker power $T^{\otimes k}$ decomposes into a direct sum of tensors $T^{w} \in [[w_1]] \otimes [[w_2]] \otimes [[w_3]]$ going over all types $(w_1, w_2, w_3) \vDash_{(m_1,m_2,m_3)} k$.
From the decomposition 
\[
T^{\otimes k} = \sum_{s_1 \in [m_1]^k} \sum_{s_2 \in [m_2]^k} \sum_{s_3 \in [m_3]^k} \prod_{\ell = 1}^k t_{s_{1\ell} s_{2\ell} s_{3\ell}} e_{s_1} \otimes e_{s_2} \otimes e_{s_3},
\]
it follows that the component $T^w$ is nonzero if and only if there is a sequence of triples $(s_{1\ell}, s_{2\ell}, s_{3\ell}) \in (\supp(T))^k$ such that the sequences $(s_{i1}, \dots, s_{ik})$ have types $w_i$.
Or, in terms of probability distributions, if and only if there is a rational probability distribution on $\supp(T)$ with common denominator $k$ such that the marginal distributions are $\frac{w_i}{k}$.
Let $W(T)$ be the polytope consisting of all triples of marginal distributions for all probability distributions on $\supp(T)$.

\begin{theorem}\label{lem:tight-ub}
  Let $T \in K^{m_1 \times m_2 \times m_3}$ be a tensor.
  For every $\xi \in \Xi$
  \[
  \SR_{\xi}(T) \leq  \max_{p \in W(T)} \min_{i\in [3]}  2^{H(p_i) / \xi_i}.
  \]
\end{theorem}
\begin{proof}
Define
\[
M_{\xi, k}(T) = \max_{\substack{w \vDash_m k \\ T^{w} \neq 0}} \min_i (\dim [[w_i]])^{1/\xi_i}.
\]
The number of possible types $w$ is at most $(k + 1)^{m_1 + m_2 + m_3}$, and by the same reasoning as in~\cref{lem:sr-polytope-upper} we have $S_{\xi}(T^{\otimes k}) \leq (k + 1)^{\gamma} M_{\xi, k}(T)$ for some $\gamma$.
Using a well-known upper bound on the multinomial coefficients in terms of entropy,
\[
M_{\xi,k}(T) \leq \max_{\substack{w \vDash_m k \\ T^{w} \neq 0}} \min_i 2^{k H(\frac{w_i}{k})/\xi_i} \leq \max_{p \in W(T)} \min_i 2^{k H(p_i)/\xi_i}.
\]
This gives the required upper bound.
\end{proof}

A tensor $T \in K^{m_1 \times m_2 \times m_3}$ is called \emph{tight} (in the standard basis) if there exist three injective maps $u_i \colon [m_i] \to \Z$ such that for every $(j_1, j_2, j_3) \in \supp(T)$ it holds that $u_1(j_1) + u_2(j_2) + u_3(j_3) = 0$ \cite{strassen1991degeneration}.

\begin{lemma}\label{lem:sr-tight-semis}
If a tensor $T \in K^{m_1 \times m_2 \times m_3}$ is tight, then all nonzero components $T^w$ are semistable.
\end{lemma}
\begin{proof}
We may assume that $K$ is algebraically closed by extending the field to the algebraic closure if necessary (since this does not affect tightness or the components $T^w$, and since semistability is defined over the closure).
Since $T$ is tight there is by definition for every $i \in [3]$ an injective map $u_i \colon \{1, \dots, m_i\} \to \Z$ such that for every $(j_1, j_2, j_3) \in \supp(T)$ it holds that $u_1(j_1) + u_2(j_2) + u_3(j_3) = 0$.

Equip the spaces $V_i = K^{m_i}$ with actions of $K^{\times}$ by letting $t \in K^{\times}$ act as $t \cdot e_j = t^{u_i(j)} e_j$.
Taking $k$ copies of these representations, we obtain actions of the algebraic torus $(K^{\times})^k$ on $V_i^{\otimes k}$.
Additionally, the symmetric group $S_k$ acts on $V_i^{\otimes k}$ by permuting the tensor factors.
These two actions combine to give actions of the semidirect product $(K^{\times})^k \rtimes S_k$ on the three spaces $V_i^{\otimes k}$.

Note that the type spaces $[[w]] \subset V_i^{\otimes k}$ are irreducible $(K^{\times})^k \rtimes S_k$-subrepresentations of $V_i^{\otimes k}$.
Indeed, as a~$(K^{\times})^k$-representation $V_i^{\otimes k}$ decomposes into a direct sum of one-dimensional representations spanned by the basis vectors $e_j$, which all have different weights $(u_i(j_1), \dots, u_i(j_k))$. (Here we are using that the $u_i$ are injective.)
Therefore, every $(K^{\times})^k \rtimes S_k$-subrepresentation is spanned by some subset of these basis vectors which is closed under permutations.
Since every two sequences of the same type can be transformed into each other by a permutation, every such subrepresentation is a sum of type spaces.

Actions of $K^{\times}$ on $V_i$ define the action on the tensor product $V_1 \otimes V_2 \otimes V_3$.
An element $t \in K^{\times}$ acts by $t \cdot e_{j_1} \otimes e_{j_2} \otimes e_{j_3} = t^{u_1(j_1) + u_2(j_2) + u_3(j_3)} e_{j_1} \otimes e_{j_2} \otimes e_{j_3}$.
Since $u_1(j_1) + u_2(j_2) + u_3(j_3) = 0$ for all triples $(j_1, j_2, j_3) \in \supp(T)$,
the tensor $T \in V_1 \otimes V_2 \otimes V_3$ is invariant under this action of $K^{\times}$.
It follows that the Kronecker power $T^{\otimes k}$ is invariant under $(K^{\times})^k$.
Since $T^{\otimes k}$ is also invariant under permutation of the $k$ factors, it is invariant under the action of $(K^{\times})^k \rtimes S_k$ on $V_1^{\otimes k} \otimes V_2^{\otimes k} \otimes V_3^{\otimes k}$.

Since $T^{\otimes k}$ is invariant under the action of $(K^{\times})^k \rtimes S_k$ and the projections $V_i^{\otimes k} \to [[w_i]]$ are equivariant, all components $T^w$ are also invariant.
Using the semistability test~(\cref{lem:irrsemis}) we conclude that if $T^w$ is nonzero, then it is semistable.
\end{proof}

\begin{lemma}\label{lem:sr-tight-lower}
  Let $T \in K^{m_1 \times m_2 \times m_3}$ be a tight tensor. For every $\xi \in \Xi$ we have $S_{\xi}(T^{\otimes k}) \geq c_{\xi} M_{\xi, k}(T)$ for some $c_{\xi} > 0$ dependent only on $\xi$.
\end{lemma}
\begin{proof}
The proof of this statement is along the same lines as the proof of \cref{lem:sr-polytope-lower}.
For every $w \vDash_{(m_1, m_2, m_3)} k$ the tensor $T^w$ is a restriction of $T^{\otimes k}$. By~\cref{lem:sr-tight-semis}, $T^w$ is semistable whenever it is nonzero, and using monotonicity of weighted slice ranks and~\cref{cor:slice-semis} we obtain
\[
S_{\xi}(T^{\otimes k}) \geq S_{\xi}(T^w) \geq c_{\xi} \min_i (\dim [[w_i]])^{1/\xi_i}.
\]
Maximizing over all $w$ such that $T^w$ is nonzero, we get the required bound.
\end{proof}

\begin{theorem}\label{thm:sr-tight}
  Let $T \in K^{m_1 \times m_2 \times m_3}$ be a tight tensor.
  For every $\xi \in \Xi$ we have
  \[
     \SR_{\xi}(T) = \max_{p \in W(T)} \min_{i\in [3]}  2^{H(p_i) / \xi_i}.
  \]
\end{theorem}
\begin{proof}
The upper bound is given by~\cref{lem:tight-ub}.
The proof of the lower bound follows the same reasoning as the proof of~\cref{lem:sr-polytope-limit} and~\cref{thm:sr-polytope}. The set $R_k(T)$ in this case is defined as $\{ w \vDash_m k \mid T^w \neq 0\}$. Instead of~\cref{lem:reprdim-sum} we use the basic inequality $\binom{k+\ell}{w+v} \geq \binom{k}{w} \binom{\ell}{v}$.
\end{proof}

\section{Minimax correspondence}\label{sec:correspondence}

In this section we introduce the notion of a \emph{dual pair}, a variation on the well-known Legendre–Fenchel transformation for convex functions which is also known as the convex conjugate. 
For functions that form a dual pair, we use the von Neumann minimax theorem to prove what we call a \emph{minimax correspondence}: each function can be computed as an optimization problem in terms of the other. The result in this section is general. We will apply it to the quantum functionals and the asymptotic weighted slice ranks, and variations, in \cref{sec:quantum-functionals}.

\subsection{Definition of a dual pair}

\begin{definition}[Dual pair]\label{def:conjugate-pair}
  Let $\Pi \subseteq \R^n$ be a compact convex set.
  Let $h_1, \dots, h_k$ be functions $\Pi \to \R_{\geq 0}$ that are continuous and concave.
  We will write $h(x) = (h_1(x), \ldots, h_k(x))$.
  Define the parameter spaces
  \begin{align*}
    \Theta &= \{ (\theta_1, \dots, \theta_k) \mid \theta_i \geq 0,\,  \sum_{i = 1}^k \theta_i = 1 \},\\
    \Xi &= \{(\xi_1, \dots, \xi_k) \mid \xi_i \geq 0,\, \max \{\xi_1, \dots, \xi_k\} = 1\}.
  \end{align*}
  Let $\langle v,w\rangle \coloneqq \sum_{i=1} v_i w_i$ denote the standard dot product on real vectors.
  We define the functions $f \colon \Theta \to \R$ and $g \colon \Xi \to \R$
  by
  \begin{align*}
    f(\theta) &= \max_{x \in \Pi}\, \langle\theta, h(x)\rangle,\\
    g(\xi) &= \max_{x \in \Pi} \min_{i \in [k]} \frac{h_i(x)}{\xi_i}.
  \end{align*}
  We call $(f,g)$ a \emph{dual pair}.
\end{definition}

Clearly the template and motivation for the above definition are pairs formed by (the logarithms of) the quantum functionals and the asymptotic weighted slice rank. Namely, for any fixed tensor~$T$ we have that the pair $(f_\theta(T), g_\xi(T))$ is a dual pair where $f_\theta(T)$ is the logarithm of the quantum functional at $T$ and $g_\xi(T)$ is the logarithm of the asymptotic $\xi$-weighted slice rank of $T$.

\begin{remark}
The notion of a dual pair is not symmetric. That is, if $(f,g)$ is a dual pair, then~$(g,f)$ is not necessarily a dual pair.
\end{remark}

\subsection{Von Neumann minimax theorem}
To prove the minimax correspondence for dual pairs we use the von Neumann minimax theorem for quasi\-concave-quasi\-convex functions. 

Let $X$ be a convex set and
let $f \colon X \to \R$ be a function. 
We define $f$ to be \emph{quasiconvex} if $f(\lambda x + (1-\lambda)y) \leq \max \{f(x), f(y)\}$ for every $\lambda \in [0,1]$. We define $f$ to be \emph{quasiconcave} if $f(\lambda x + (1-\lambda)y) \geq \max \{f(x), f(y)\}$ for every $\lambda \in [0,1]$.
Equivalently, $f$ is quasiconvex if the sublevel set $\{x \mid f(x) \leq c\}$ is convex for every $c$, and $f$ is quasiconcave if the superlevel set $\{x \mid f(x) \geq c\}$ is convex for every~$c$. 
Convex functions are quasi-convex and concave functions are quasi-concave.

We define $f$ to be \emph{quasilinear} if it is both quasiconvex and quasiconcave.
Every fractional linear function $f(x) = \langle u,x \rangle/\langle v,x\rangle$
is a quasilinear function on any convex set on which $\langle v, x\rangle$ does not change its sign.
To see this, note that $\langle u,x\rangle /\langle v,x\rangle \leq c$ if and only if $\langle u - cv, x\rangle \leq 0$ in case when $\left<v, x\right>$ is positive and similarly when $\leq$ is replaced by $\geq$. Thus the sublevel and superlevel sets are halfspaces, and in particular they are convex.
However, generally a fractional linear function is neither convex nor concave.

Von Neumann proved the following minimax theorem for quasi\-concave-quasi\-convex functions in 1937, extending the earlier version from 1928 which was for bilinear functions. See, for instance,~\cite[Theorem~2, page~3]{Simons1995}.

\begin{theorem}\label{th:minimax}
Let $X \subset \R^m$ and $Y \subset \R^n$ be nonempty compact convex sets and let $f \colon X \times Y \to \R$ be a continuous function such that $f(x, \cdot)$ is quasiconvex for every $x \in X$ and $f(\cdot, y)$ is quasiconcave for every $y \in Y$.
Then
\[
\max_{x\in X} \min_{y\in Y} f(x, y) = \min_{y \in Y} \max_{x\in X} f(x, y).
\]
\end{theorem}

We will also use the following property of quasiconvex functions.

\begin{lemma}\label{claim:opt-fractlinear}
  Let $\Pi \subseteq \R^m$ be a convex polytope and $f \colon \Pi \to \R$ be a quasiconcave function. Then the minimal value of $f$ is attained at some vertex of~$\Pi$.
\end{lemma}
\begin{proof}
Let $c$ be minimal among the values of $f$ on all vertices of $\Pi$.
Since the set $\{x \mid f(x) \geq c\}$ is convex and contains all the vertices of $\Pi$, it contains its convex hull $\Pi$.
Thus, $f$ does not attain values that are strictly less than $c$.
\end{proof}

In particular, it follows from \cref{claim:opt-fractlinear} that every fractional linear function $\langle u,x \rangle/\langle v,x\rangle$ defined on a convex polytope attains its maxima and minima at vertices of the convex polytope. 

\subsection{Minimax correspondence}
We will now show using the minimax theorem that there is a correspondence between functions that form a dual pair.

For $\xi \in \Xi$ define the restricted parameter space
\[
\Theta(\xi) = \{ \theta \in \Theta \mid \xi_i =0 \Rightarrow \theta_i = 0\}.
\]
Note that if $\theta \in \Theta(\xi)$, then $\left<\theta, \xi\right>$ is nonzero.

\begin{theorem}[Minimax correspondence]\label{legendrevar}
  Let $(f,g)$ be a dual pair.
  Then $f$ and $g$ are continuous and $f$ is convex.
  Moreover, $f$ and $g$ are expressed in terms of each other via
  \begin{align*}
    g(\xi) & = \min_{\theta \in \Theta(\xi)} \frac{f(\theta)}{\langle\theta, \xi\rangle},\\ 
    f(\theta) &= \max_{\xi \in \Xi}\, g(\xi) \langle\theta, \xi\rangle. 
  \end{align*}
\end{theorem}
\begin{proof}
  First of all, the function $f(\theta)$ is continuous and convex since it is the pointwise maximum of the equicontinuous and
  convex functions $\theta \mapsto \langle\theta, h(x)\rangle$.

  Next, we will prove that $g(\xi) = \min_{\theta \in \Theta(\xi)} f(\theta) / \langle\theta, \xi\rangle$.
  By definition we have
  \begin{equation}\label{eq:g}
  g(\xi) = \max_{x \in \Pi} \min_{i \in [k]} \frac{h_i(x)}{\xi_i} = \max_{x \in \Pi} \min_{i \in [k], \xi_i \neq 0} \frac{h_i(x)}{\xi_i}.
  \end{equation}
  From applying \cref{claim:opt-fractlinear} to the quasiconcave function $\theta \mapsto \langle\theta, h(x)\rangle/\langle\theta, \xi\rangle$ on the simplex $\Theta(\xi)$ we get the equation
  \[
  \min_{i \in [k], \xi_i \neq 0} \frac{h_i(x)}{\xi_i} = \min_{\theta \in \Theta(\xi)} \frac{\langle\theta, h(x)\rangle}{\langle\theta, \xi\rangle}
  \]
  and so
  \[
  g(\xi) = \max_{x \in \Pi} \min_{\theta \in \Theta(\xi)} \frac{\langle\theta, h(x)\rangle}{\langle\theta, \xi\rangle}.
  \]
  Then applying the minimax theorem (\cref{th:minimax}) gives
  \[
    g(\xi)
    = \min_{\theta \in \Theta(\xi)} \max_{x \in \Pi}
    \frac{\langle\theta, h(x)\rangle}{\langle\theta, \xi\rangle}
    = \min_{\theta \in \Theta(\xi)}
    \frac{f(\theta)}{\langle\theta, \xi\rangle}.
  \]
  This proves the claim. It also shows that $g(\xi)$ is continuous, since it is the 
  pointwise minimum of the equicontinuous family of continuous functions $\theta \mapsto f(\theta)/\langle \theta, \xi\rangle$.

  We will now prove $f(\theta) = \max_{\xi \in \Xi}\, g(\xi) \langle\theta, \xi\rangle$.
  
  Plugging \eqref{eq:g} into $\max_{\xi \in \Xi}\, g(\xi) \langle\theta, \xi\rangle$ we get
  \[
  \max_{\xi \in \Xi}\, g(\xi) \langle\theta, \xi\rangle
    = \max_{\xi \in \Xi} \max_{x \in \Pi} \min_{i \in [k], \xi_i \neq 0}
    \frac{h_i(x)}{\xi_i} \langle\theta, \xi\rangle.
  \]
  Exchanging the maximizations gives
  \[
    \max_{\xi \in \Xi}\, g(\xi) \langle\theta, \xi\rangle
    = \max_{x \in \Pi} \max_{\xi \in \Xi} \min_{i \in [k], \xi_i \neq 0}
    \frac{h_i(x)}{\xi_i} \langle\theta, \xi\rangle.
  \]
  
  Note that for each $j$ such that $\xi_j \neq 0$ we have 
  \[
  \min_{i \in [k], \xi_i \neq 0} \frac{h_i(x)}{\xi_i} \xi_j
  \leq \frac{h_j(x)}{\xi_j} \xi_j = h_j(x),
  \]
  and if $\xi_j = 0$, then the same inequality $\min_{i \in [k], \xi_i \neq 0} \frac{h_i(x)}{\xi_i} \xi_j \leq h_j(x)$ holds because $h_j(x)$ is nonnegative.
  Therefore
  \[
    \min_{i \in [k], \xi_i \neq 0} \frac{h_i(x)}{\xi_i} \langle\theta, \xi\rangle
    = \sum_{j = 1}^k \left(\theta_j \cdot \min_{i \in [k], \xi_i \neq 0} \frac{h_i(x)}{\xi_i} \xi_j\right)
    \leq \sum_{j = 1}^k \theta_j h_j(x) = \left<\theta, h(x)\right>.
  \]
  This upper bound is attained if $\xi$ is proportional to $h(x)$, that is, $\xi = \frac{h(x)}{\max_{i \in [k]} h_i(x)}$, or arbitrary $\xi$ if $h(x) = 0$.
  Thus
  \[
    \max_{\xi \in \Xi} \min_{i \in [k], \xi_i \neq 0}
    \frac{h_i(x)}{\xi_i} \langle\theta, \xi\rangle = \left<\theta, h(x)\right>
  \]
  and
  \[
    \max_{\xi \in \Xi}\, g(\xi) \langle\theta, \xi\rangle
    = \max_{x \in \Pi} \max_{\xi \in \Xi} \min_{i \in [k], \xi_i \neq 0}
    \frac{h_i(x)}{\xi_i} \langle\theta, \xi\rangle
    = \max_{x \in \Pi} \langle \theta, h(x) \rangle
    = f(\theta).
  \]
  We conclude that the claim holds.
\end{proof}

\subsection{Preordered semirings and the minimax correspondence}
We extend the ideas of the previous section to families of dual pairs that are parametrized by elements of a preordered semiring. This is mostly a notational extension. In our application, the preordered semiring will be the set of tensors with multiplication given by tensor product, addition given by direct sum and preorder given by restriction.  We will then discuss how the parametrization behaves under the minimax correspondence. 

\begin{definition}
Let $T \mapsto \Pi(T)\subseteq \R^n$ be a map that assigns a compact convex set to every element~$T$ from some index set that has defined on it a product $\otimes$, a  sum $\oplus$ and a preorder~$\leq$.\footnote{In our applications this index set will be the set of $k$-tensors, with product given by the tensor product, sum given by the direct sum and preorder given by restriction.} Let $h_1, \ldots, h_k$ be functions $\cup_T \Pi(T) \to \R_{\geq 0}$ that are continuous and quasiconcave, as before. For every $\theta \in \Theta$ and $\xi \in \Xi$ we define the functions
\begin{align*}
f_\theta(T) &= \max_{x \in \Pi(T)} \langle \theta, h(x)\rangle,\\
g_\xi(T) &= \max_{x \in \Pi(T)} \min_{i \in [k]} \frac{h_i(x)}{\xi_i}.
\end{align*}
Note that compared to the notation for dual pairs, in our new notation the parameters $\theta$ and $\xi$ have been moved into the subscript. 
For every $T$ the functions $\xi \mapsto g_\xi(T)$ and $\theta \mapsto f_\theta(T)$ form a dual pair. We say that $(f_\theta, g_\xi)$ is a \emph{dual pair family}.
\end{definition}

Again our main example is given by the quantum functionals and the asymptotic weighted slice ranks: they form a dual pair family.

If $(f_\theta, g_\xi)$ is a dual pair family, then by the minimax correspondence (\cref{legendrevar}) we know that~$f_\theta(T)$ is continuous in $\theta$, $g_\xi(T)$ is continuous in $\xi$, and $f_\theta(T)$ is convex in $\theta$.
Moreover, $f_\theta(T)$ and $g_\xi(T)$ are related via
\begin{align*}
g_\xi(T) & = \min_{\theta \in \Theta(\xi)} \frac{f_\theta(T)}{\langle\theta, \xi\rangle},\\
f_\theta(T) &= \max_{\xi \in \Xi}\, g_\xi(T) \langle\theta, \xi\rangle .
\end{align*}
We define $F_\theta(T) = 2^{f_\theta(T)}$ and $G_\xi(T) = 2^{g_\xi(T)}$. 
The minimax correspondence between $g_\xi(T)$ and~$f_\theta(T)$ allows us to transfer properties:

\begin{theorem}\label{conseq1}
Let $S$ and $T$ be fixed elements from the index set.
\begin{enumerate}[label=\upshape(\roman*)]
    \item\label{item:super-mult} If $F_\theta(S\otimes T) \geq F_\theta(S)F_\theta(T)$ for all $\theta$, then $G_\xi(S\otimes T) \geq G_\xi(S) G_\xi(T)$ for all $\xi$.
    \item\label{item:super-add} If $F_\theta(S\oplus T) \geq F_\theta(S) + F_\theta(T)$ for all $\theta$, then $G_\xi(S\oplus T)\geq G_\xi(S) + G_\xi(T)$ for all $\xi$.
    \item If $G_\xi(S\otimes T) \leq G_\xi(S) G_\xi(T)$ for all $\xi$, then $F_\theta(S\otimes T) \leq F_\theta(S) F_\theta(T)$ for all $\theta$.
    \item If $G_\xi(S\oplus T)\leq G_\xi(S) + G_\xi(T)$ for all $\xi$, then $F_\theta(S\oplus T) \leq F_\theta(S) + F_\theta(T)$ for all $\theta$.
    \item $F_\theta(S) \geq F_\theta(T)$ for all $\theta$ if and only if $G_\xi(S) \geq G_\xi(T)$ for all $\xi$.
    \item $F_\theta(T) = \alpha$ for all $\theta$ if and only if $G_\xi(T) = \alpha$ for all $\xi$.
\end{enumerate}
\end{theorem}
\begin{proof}
  \ref{item:super-mult} We will show $G_\xi(S \otimes T) \geq G_\xi(S) G_\xi(T)$ by showing that $g_\xi(T_1 \otimes T_2) \geq g_\xi(T_1) + g_\xi(T_2)$. We apply the minimax correspondence and then the assumption $f_\theta(S \otimes T) \geq f_\theta(S) + f_\theta(T)$ to get
  \[
    \sr_{\xi}(S \otimes T)
    = \min_{\theta \in \Theta(\xi)} \frac{f_{\theta}(S \otimes T)}{\langle \theta, \xi\rangle}
    \geq \min_{\theta \in \Theta(\xi)} \frac{f_{\theta}(S) + f_{\theta}(T)}{\langle \theta, \xi\rangle}
  \]
    Separating the minimization over $\theta$ into two parts and then applying the minimax correspondence again, we get
  \begin{multline*}
    \min_{\theta \in \Theta(\xi)} \frac{f_{\theta}(S) + f_{\theta}(T)}{\langle \theta, \xi\rangle} = \min_{\theta \in \Theta(\xi)} \left[\frac{f_{\theta}(S)}{\langle \theta, \xi \rangle}
    + \frac{f_{\theta}(T)}{\langle \theta, \xi\rangle} \right] \\
    \geq \min_{\theta \in \Theta(\xi)} \frac{f_{\theta}(S)}{\langle \theta, \xi \rangle}
    + \min_{\theta \in \Theta(\xi)} \frac{f_{\theta}(T)}{\langle \theta, \xi \rangle}
    = \sr_{\xi}(S) + \sr_{\xi}(T).
  \end{multline*}
  This proves the claim.
  
  \ref{item:super-add} For this proof we use the general fact that $x + y = \max_{0 \leq \eps \leq 1} x^\eps y^{1-\eps} 2^{h(\eps)}$ for any nonnegative reals $x,y$ \cite[Eq.~2.13]{strassen1991degeneration}, where $h$ denotes the binary entropy as usual. We know that $F_\theta(T_1 \oplus T_2) \geq F_\theta(T_1) + F_\theta(T_2)$ and so
\begin{align*}
\SR_\xi(T_1 \oplus T_2) &= \min_{\theta \in \Theta(\xi)} F_\theta(T_1 \oplus T_2)^{\tfrac{1}{\langle \theta, \xi\rangle}} \geq \min_{\theta \in \Theta(\xi)} ( F_\theta(T_1) + F_\theta(T_2))^{\tfrac{1}{\langle \theta, \xi\rangle}}.
\end{align*}
From the aforementioned general fact it follows that
\[
\min_{\theta \in \Theta(\xi)} ( F_\theta(T_1) + F_\theta(T_2))^{\tfrac{1}{\langle \theta, \xi\rangle}}
= \min_{\theta \in \Theta(\xi)} \max_{0 \leq \eps \leq 1} (F_\theta(T_1)^\eps F_\theta(T_2)^{1-\eps} 2^{h(\eps)})^{\tfrac{1}{\langle \theta, \xi\rangle}}.
\]
Now we apply the von Neumann minimax theorem (after taking the logarithm) to get that 
\begin{align*}
\min_{\theta \in \Theta(\xi)} \max_{0 \leq \eps \leq 1} (F_\theta(T_1)^\eps F_\theta(T_2)^{1-\eps} 2^{h(\eps)})^{\tfrac{1}{\langle \theta, \xi\rangle}}
&= \max_{0 \leq \eps \leq 1} \min_{\theta \in \Theta(\xi)} (F_\theta(T_1)^\eps F_\theta(T_2)^{1-\eps} 2^{h(\eps)})^{\tfrac{1}{\langle \theta, \xi\rangle}}.
\end{align*}
We pull the minimization over $\theta$ inside the product and then apply the minimax correspondence to get
\begin{align*}
  &\phantom{=}\max_{0 \leq \eps \leq 1} \min_\theta (F_\theta(T_1)^\eps F_\theta(T_2)^{1-\eps} 2^{h(\eps)})^{\tfrac{1}{\langle \theta, \xi\rangle}} \\&\geq \max_{0 \leq \eps \leq 1} \min_\theta (F_\theta(T_1)^\eps)^{\tfrac{1}{\langle \theta, \xi\rangle}} \min_\theta(F_\theta(T_2)^{1-\eps})^{\tfrac{1}{\langle \theta, \xi\rangle}} \min_\theta(2^{h(\eps)})^{\tfrac{1}{\langle \theta, \xi\rangle}}\\
&= \max_{0 \leq \eps \leq 1} \SR_\xi(T_1)^\eps \SR_\xi(T_2)^{1-\eps}  \min_\theta (2^{h(\eps)})^{\tfrac{1}{\langle \theta, \xi\rangle}}\\
&\geq \max_{0 \leq \eps \leq 1} \SR_\xi(T_1)^\eps \SR_\xi(T_2)^{1-\eps}   2^{h(\eps)}.
\end{align*}
Another application of $x + y = \max_{0 \leq \eps \leq 1} x^\eps y^{1-\eps} 2^{h(\eps)}$ gives
\begin{align*}
\max_{0 \leq \eps \leq 1} \SR_\xi(T_1)^\eps \SR_\xi(T_2)^{1-\eps}   2^{h(\eps)}
&= \SR_\xi(T_1) + \SR_\xi(T_2).
\end{align*}
This proves the claim.

The proofs for the other claims of the theorem go similarly.
\end{proof}

\section{Quantum functionals over arbitrary fields}\label{sec:quantum-functionals}

The quantum functionals are a family of functions on complex tensors. We discussed their special properties already in the introduction and will recall them now. For $\theta \in \Theta$ and any complex tensor~$T$ the quantum functionals are defined by 
\[
F_\theta(T) \coloneqq \max_{p \in \Pi(T)} 2^{\langle \theta, H(p)\rangle}.
\]
Their key properties~\cite{DBLP:conf/stoc/ChristandlVZ18,cvzjournal} are
\begin{itemize}
\item
Monotonicity under restriction: $S \leq T \Rightarrow F_\theta(S) \leq F_\theta(T)$,
\item
Additivity under direct sum: $F_{\theta}(S \oplus T) = F_{\theta}(S) + F_{\theta}(T)$,
\item
Multiplicativity under tensor product: $F_{\theta}(S \otimes T) = F_{\theta}(S) \cdot F_{\theta}(T)$,
\item
Normalization on unit tensors: $F_{\theta}(I_n) = n$ for any $n \in \N$ where $I_n = \sum_{i=1}^n e_i \otimes e_i \otimes e_i$. 
\end{itemize}
Functions with the above properties are called \emph{universal spectral points} \cite{strassen1988asymptotic}. Such universal spectral points are powerful tools for the study of, for example, barriers for fast matrix multiplication~\cite{christandl2020barriers}, and other question about large tensor powers of tensors.

The above construction of universal spectral points does not allow immediate extension to tensors over other fields than the complex numbers. The purpose of this section is to construct a family of functionals on tensors over arbitrary fields which we conjecture to be universal spectral points, and give evidence for the validity of this conjecture.

The construction is as follows. Fix some field~$K$ and let $T$ be a tensor over~$K$.
Let $\SR_\xi(T)$ denote the asymptotic $\xi$-weighted slice rank of $T$. 
We define the \emph{quantum functionals over~$K$} by 
\[
F^K_{\theta}(T) \coloneqq \max_{\xi \in \Xi} \SR_{\xi}(T)^{\langle \theta, \xi\rangle}.
\]

\begin{proposition}\label{prop:mon-norm}
 The $F^K_{\theta}$ are monotone and normalized.
\end{proposition}
\begin{proof}
We know that the weighted slice rank $S_\xi$ is monotone and normalized (\cref{subsec:basic-prop}). Therefore, the same holds for the asymptotic weighted slice rank $G_\xi$, and in turn the same holds for the function~$F_\theta^K$.
\end{proof}
\begin{conjecture}\label{conj:main}
 The $F^K_{\theta}$ are additive and multiplicative. 
\end{conjecture}
\cref{prop:mon-norm} and \cref{conj:main} (if true) combined say that the $F_\theta^K$ are universal spectral points for tensors over $K$.

We discuss two pieces of evidence for \cref{conj:main} in the following subsections. The first is that the conjecture is true over the complex numbers, since then $F^\C_\theta$ coincides with the quantum functionals $F_\theta$ from~\cite{DBLP:conf/stoc/ChristandlVZ18,cvzjournal}. The second is that the conjecture is true for the subclass of tight tensors over arbitrary fields, since then $F^K_\theta$ coincides with a family of tensor parameters called the support functionals~\cite{strassen1991degeneration}. Both proofs will use the minimax correspondence of~\cref{sec:correspondence}.

Finally, in \cref{subsec:min-entropy}, as another application of the minimax correspondence, we introduce and discuss the min-entropy quantum functionals.

\subsection{Arbitrary tensors over the complex numbers}

\begin{theorem}\label{thm:f-sr}
  The quantum functionals and asymptotic slice ranks are related as follows.
  \[
    \SR_{\xi}(T) = \min_{\theta \in \Theta(\xi)}\, F_{\theta}(T)^{1/\langle \theta, \xi\rangle},\quad
    F_{\theta}(T) = \max_{\xi \in \Xi}\, \SR_{\xi}(T)^{\langle \theta, \xi\rangle}.
  \]
\end{theorem}
\begin{proof}
We will use \cref{legendrevar}.
Let $\Pi$ be the moment polytope of $T$. For $p \in \Pi$ let $h_i(p) = H(p_i)$ be the Shannon entropy of the $i$th component of $p$. Then $h_i(p)$ is a continuous and concave function. We see that $f(\theta)$ is the logarithmic quantum functional and by Theorem \ref{thm:sr-polytope} we see that $g(\xi)$ is the asymptotic $\xi$-weighted slice rank. \cref{legendrevar} proves the claim.
\end{proof}
In particular, \cref{thm:f-sr} implies that \cref{conj:main} holds for $K = \C$.

Moreover, from \cref{thm:f-sr} and the general \cref{conseq1} it follows that:

\begin{corollary}\label{cor:sr-super}
  The asymptotic $\xi$-weighted slice rank $\SR_\xi$ on complex tensors is monotone, normalized, super-multiplicative, and super-additive.
\end{corollary}

\subsection{Tight tensors over arbitrary fields}

Let $K$ be an arbitrary field. Recall that a tensor $T \in K^{m_1 \times m_2 \times m_3}$ is called \emph{tight} in the standard bases if there exist injective maps $u_i \colon [m_i] \to \Z$ such that the sum $u_1(j_1) + u_2(j_2) + u_3(j_3)$ is constant when we vary over all $(j_1, j_2, j_3) \in \supp(T)$.
For tight tensors over an arbitrary field 
the support functionals of \cite{strassen1991degeneration} are defined as
\[
\zeta_{\theta}(T) = \max_{p \in W(T)} 2^{\langle\theta, H(p)\rangle},
\]
where $W(T)$ is the polytope defined in \cref{subsec:tight}.
For tight tensors these functionals are spectral points, that is, they are monotone, multiplicative, additive and normalised \cite{strassen1991degeneration}\footnote{In fact, this is true for the slightly larger class of \emph{oblique} tensors. We refer to \cite{strassen1991degeneration} and \cite{cvzjournal} for details.}.

\begin{theorem}\label{thm:f-sr-tight}
  Let $T$ be a tight tensor.
  The support functionals and asymptotic slice ranks of $T$ are related as follows:
  \[
    \SR_{\xi}(T) = \min_{\theta \in \Theta(\xi)}\, \zeta_{\theta}(T)^{1/\langle \theta, \xi\rangle},\quad
    \zeta_{\theta}(T) = \max_{\xi \in \Xi}\, \SR_{\xi}(T)^{\langle \theta, \xi\rangle}.
  \]
\end{theorem}
\begin{proof}
The characterization of the asymptotic slice ranks of tight tensors proven in \cref{thm:sr-tight} shows that $(\log \zeta_{\theta}, g_{\xi})$ form a dual pair family. The claim follows from the minimax correspondence (\cref{legendrevar}).
\end{proof}

In particular, \cref{thm:f-sr-tight} implies that \cref{conj:main} holds for all tight tensors over an arbitrary field $K$.

Moreover, from \cref{thm:f-sr-tight} and the general \cref{conseq1} it follows that:

\begin{corollary} 
\label{corx}
  The asymptotic $\xi$-weighted slice rank $\SR_\xi$ on tight tensors over arbitrary fields is monotone, normalized, super-multiplicative, and super-additive.
\end{corollary}
We do not know if there is a direct proof of \cref{corx} not involving the connection to the Strassen's functionals.

\subsection{Min-entropy quantum functionals}\label{subsec:min-entropy}

Let $F_{\theta, \infty}$ be the min-entropy quantum functionals as defined in \cref{subsec:intro-min-entropy}.

\begin{theorem}\label{th:min-entropy-basic}
The min-entropy quantum functionals are monotone, normalized and super-mul\-ti\-plica\-tive.
\end{theorem}
\begin{proof}
Monotonicity follows from the fact that if $S \leq T$, then $\Pi(S) \subseteq \Pi(T)$. For the normalization, for every $r \in \N$ it is clear that $F_{\theta, \infty}(I_r) \leq F_\theta(I_r) = r$. On the other hand, it follows directly from the spectral definition of the moment polytope that the flattenings of $I_r$ in the standard basis have uniform singular values and so $F_{\theta, \infty}(I_r) \geq r$. Super-multiplicativity follows from the characterization (see Equation~\eqref{eq:charac}) of the min-entropy quantum functionals as
\[
F_{\theta,\infty}(T) = \sup_{S \in G\cdot T} \prod_{i=1}^3 \Biggl( \frac{ \norm[0]{S_{(i)}}_F^2}{\norm[0]{S_{(i)}}_2^2} \Biggr)^{\!\!\theta_i},
\]
and the fact that the Frobenius norm and spectral norm are multiplicative under the Kronecker product.
\end{proof}

Recall that in the introduction we defined, for $\xi \in \Xi$, the function
\[
    G_{\xi, \infty}(T) = \max_{p \in \Pi(T)} \min_{i\in [3]}  2^{H_\infty(p_i) / \xi_i}.
\]
Since the min-entropy $H_\infty$ is continuous and concave, we find (analogous to the proof of \cref{thm:f-sr}) via the minimax correspondence (\cref{legendrevar}) that:

\begin{theorem}\label{thm:m-f-sr}
For every $\xi \in \Xi$ and every tensor $T$ we have
\[
    \SR_{\xi, \infty}(T) = \min_{\theta \in \Theta(\xi)} {F_{\theta, \infty}(T)}^{1/\langle \theta, \xi\rangle}.
\]
For every $\theta \in \Theta$ and every tensor $T$ we have
\[
    F_{\theta, \infty}(T) = \max_{\xi \in \Xi} \SR_{\xi, \infty}(T)^{\langle \theta, \xi\rangle}.
\]
\end{theorem}

Again, via the general \cref{conseq1} we find that:

\begin{corollary}\label{cor:m-sr-super}
  For every $\xi \in \Xi$ the function $\SR_{\xi, \infty}$ is monotone, normalized and super-mul\-ti\-plica\-tive.
\end{corollary}

The function $\SR_{\xi, \infty}(T)$ is similar to the G-stable rank that was introduced in \cite{derksen2020gstable}. However, our weighting is slightly different from the weighting using in the definition of G-stable rank.

\appendix

\section{Invariance, polystability and semistability}\label{sec:sstest}

We will prove the semistability test given in \cref{lem:irrsemis}. 
This lemma is essentially an adaptation of a result proven by Kempf~\cite{kempf1978instability} to the tensor action of the group $\SL(V_1) \times \SL(V_2) \times \SL(V_3)$.

The result of Kempf actually deals with a notion that is stronger than semistability, called polystability.
Given a representation $V$ of a reductive group $G$, any nonzero element $v \in V$ is called \emph{polystable} if the orbit~$G\cdot v$ is Zariski closed.
Clearly, polystable implies semistable.
When working in projective space, we call an element $[v] \in \bbP V$ polystable if and only if $v \in V$ is polystable.

Kempf proved the following result:

\begin{theorem}[{Kempf's polystability test~\cite[Cor.~5.1]{kempf1978instability}}]\label{th:kempf}
Let $V$ be a representation of a reductive group $G$.
Then $G$ acts on the projective space $\bbP V$ of lines in $V$.
Assume that $G$ has no nontrivial central $1$-parameter subgroup. If the stabilizer of an element $y \in \bbP V$ in $G$ is not contained in any proper parabolic subgroup of $G$, then $y$ is polystable under $G$.
\end{theorem}

We now apply Kempf's theorem to tensors. Let $V = V_1 \otimes V_2 \otimes V_3$ and let the (reductive) group $G = \SL(V_1) \times \SL(V_2) \times \SL(V_3)$ act in the usual way on $V$. We call a nonzero tensor $T \in V_1 \otimes V_2 \otimes V_3$ \emph{polystable} if the orbit $G \cdot T$ is closed for the group $G = \SL(V_1) \times \SL(V_2) \times \SL(V_3)$.
Since polystability implies semistability, the following lemma implies~\cref{lem:irrsemis}.

\begin{lemma}\label{lem:irrpolys}
Let $\Gamma$ be a group and let $V_1, V_2, V_3$ be irreducible representations of $\Gamma$. Consider the induced action of $\Gamma$ on the tensor product $V_1 \otimes V_2 \otimes V_3$.
Every nonzero tensor $T \in V_1 \otimes V_2 \otimes V_3$ invariant under $\Gamma$ is polystable.
\end{lemma}
\begin{proof}
The center of $\SL_n$ is the finite subgroup $\{\alpha\,  \id \mid \alpha \in K,\, \alpha^n = 1\}$. Therefore, the center of $G = \SL(V_1) \times \SL(V_2) \times \SL(V_3)$ is also a finite group and thus does not contain a nontrivial $1$-parameter subgroup.

For any vector space $V$, the parabolic subgroups of $\SL(V)$ correspond to flags in $V$, in the following way.
If $\{0\} = U_0 \subset U_1 \subset \dots \subset U_k = V$ is a flag, then the corresponding parabolic subgroup consists of all maps in $\SL(V)$ that preserve the flag (that is, leave each space $U_i$ invariant).
Each proper parabolic subgroup is therefore contained in the subgroup $P_U$ which leaves some nonzero proper subspace $U$ invariant (these are maximal proper parabolic subgroups).
The parabolic subgroups of the product group $\SL(V_1) \times \SL(V_2) \times \SL(V_3)$ are precisely the products of parabolic subgroups of the $\SL(V_i)$.
Therefore, each proper parabolic subgroup of $\SL(V_1) \times \SL(V_2) \times \SL(V_3)$ leaves invariant a nonzero proper subspace $U_i \subset V_i$ in one of the three factors.

Each element $g \in \Gamma$ acts on $V_i$ via some linear map $\rho_i(g) \in \GL(V_i)$.
We can scale these linear maps to obtain elements
$\hat\rho_i(g) \coloneqq (\det \rho_i(g))^{-1/{\dim V_i}}\, \rho_i(g) \in \SL(V_i)$.
Note that for every group element $g \in \Gamma$ the corresponding triple $(\hat\rho_1(g), \hat\rho_2(g), \hat\rho_3(g)) \in \SL(V_1) \times \SL(V_2) \times \SL(V_3)$ leaves the line $[T] \in \bbP(V_1 \otimes V_2 \otimes V_3)$ invariant, because $T$ is invariant under $\Gamma$.

Since the representations $V_1$, $V_2$ and $V_3$ are irreducible, the only subspaces $U_i \subseteq V_i$ that are invariant under $\rho_i(g)$ for all $g \in \Gamma$ are $V_i$ itself and the zero subspace. This means that there is no proper parabolic subgroup of $\SL(V_1) \times \SL(V_2) \times \SL(V_3)$ containing $\{(\hat\rho_1(g), \hat\rho_2(g), \hat\rho_3(g)) \mid g \in \Gamma\}$.
Using Kempf's polystability test (\cref{th:kempf}) we conclude that $T$ is polystable.
\end{proof}

\paragraph{Acknowledgements}
MC and VL acknowledge financial support from the European Research Council (ERC Grant Agreement No.\ 818761) and VILLUM FONDEN via the QMATH Centre of Excellence (Grant No.\ 10059).
MC acknowledges financial support from the Novo Nordisk Foundation (grant NNF20OC0059939 “Quantum for Life”).
JZ was supported by a Simons Junior Fellowship and NWO Veni grant VI.Veni.212.284.
JZ thanks Visu Makam, Michael Walter, Yinan Li and Harold Nieuwboer for helpful discussions.
We thank Fabien Pazuki for the help with the French translation of the abstract.
The authors are grateful to the anonymous referees for helpful comments.


\bibliographystyle{elsarticle-num}
\bibliography{jmpa-22-3550.bib}

\end{document}